\PassOptionsToPackage{numbers, compress}{natbib}
\documentclass{article}
\newcommand{\treportflag}{1} 

\ifnum\treportflag=1
  \usepackage[preprint]{neurips_2026}
\else
  \usepackage{neurips_2026}
\fi

\usepackage[T1]{fontenc}    
\usepackage{hyperref}       
\usepackage{url}            
\usepackage{booktabs}       
\usepackage{amsfonts}       
\usepackage{nicefrac}       
\usepackage{xcolor}         
\usepackage{colortbl}
\usepackage{drago}
\usepackage{setspace}
\usepackage[normalem]{ulem}
\usepackage[on]{editing}
\usepackage{tcolorbox}
\usepackage{cleveref}

\newcommand{\M}{\mathcal{M}}
\newcommand{\Mmarg}{\M^{\mathrm{marg}}}
\newcommand{\Mzmatch}{\M^{Z\text{-match}}}
\newcommand{\Mpimatch}{\M^{\pi\text{-match}}}
\newcommand{\Mzdom}{\M^{Z\text{-dom}}}
\newcommand{\Mpidom}{\M^{\pi\text{-dom}}}
\newcommand{\eps}{\varepsilon}

\newif\ifexamples
\examplestrue  

\newcommand{\ttl}{Confounder Detection via Treatment Intent:\\ A New Observational Study Design}

\title{\ttl{}}

\author{%
  Drago Ple{\v c}ko \\
  Department of Statistics \& Data Science \\
  UCLA
  \And
  Patrik Okanovi{\' c}\\
  Department of Computer Science \\
  ETH Zurich
  \And
  Torsten Hoefler \\
  Department of Computer Science \\
  ETH Zurich
  \And
  Elias Bareinboim\\
  Causal AI Lab \\
  Columbia University
}

\usepackage{selectp}

\begin{document}

\maketitle

\begin{abstract}
  Understanding the effects of interventions is central to scientific progress, with randomized controlled trials (RCTs) regarded as the gold standard for causal inference in many applied fields. However, RCTs are costly, time-consuming, and often constrained by ethical or practical limitations, motivating the need for causal methods able to draw conclusions from observational data. While such data is collected at ever larger scale, making its use for causal inference is often hindered by the fact that not all variables affecting treatment allocation and the outcome are observed -- an issue known as unobserved confounding. In this paper, we introduce a new study design called \emph{confounder detection via treatment intent}. The idea is to query a human expert who makes treatment decisions, and ask them to compare pairs of units proposed by a principled matching strategy, with the goal of eliciting unobserved variables that explain why treatment decisions differ.
We provide a theoretical basis for such a procedure, ascertaining conditions under which such a study design may elicit unobserved confounders. 
Building on this newly established foundations, we study treatment effects of interventions in the intensive care unit (ICU). First, we show empirical evidence strongly indicating that electronic health records (EHRs) collected in ICUs are subject to unobserved confounding. By using clinical text notes as a proxy for physicians’ knowledge and leveraging natural language processing, we provide a proof of concept for our methodology in a semi-synthetic environment with a known ground truth.
\end{abstract}

\section{Introduction} \label{sec:intro}
Observational data is collected at ever larger scale across empirical
sciences, offering a valuable and inexpensive resource for answering
scientific questions \citep{hersh2013caveats, casey2013use}. This abundance of data also creates an opportunity to draw causal conclusions without resorting to costly experimentation. 
Such inference, however, does not come for free: moving across the layers of Pearl's Causal Hierarchy (PCH) \citep{bareinboim2020on} requires appropriate causal assumptions.  
We focus on the classical setting of treatment effect estimation, where the goal is to infer an interventional (Layer~2) quantity from observational data (Layer~1). The treatment is denoted by
$X$, outcome by $Y$, and a vector of confounders by $Z$, where $Z$ causally precedes $(X, Y)$. 
One of the most common causal assumptions used for identifying treatment effects in such a setting is the back-door criterion \citep{pearl:2k}, which permits adjustment-based identification of interventional quantities from observational data \citep{pearl:2k, bareinboim2025causalai}.

The assumptions required for back-door admissibility, however, need not always hold true. For instance, when some of the common causes of $X$ and $Y$ remain
unobserved -- a situation referred to as unobserved confounding -- back-door admissibility does not hold. 
Such confounding is arguably the most common roadblock faced by causal analysts working with observational data, appearing across empirical sciences, ranging from medicine and epidemiology to economics and the social sciences. 
Importantly, unobserved confounding can rarely be ruled out from the data alone. 
In this paper, we use an applied example to illustrate our methodology, specifically in the context of treatment effect estimation in the intensive care unit (ICU), where interventions are allocated by physicians attending to patients.
\begin{wrapfigure}{r}{0.42\textwidth}
    \centering
    \includegraphics[width=0.4\textwidth]{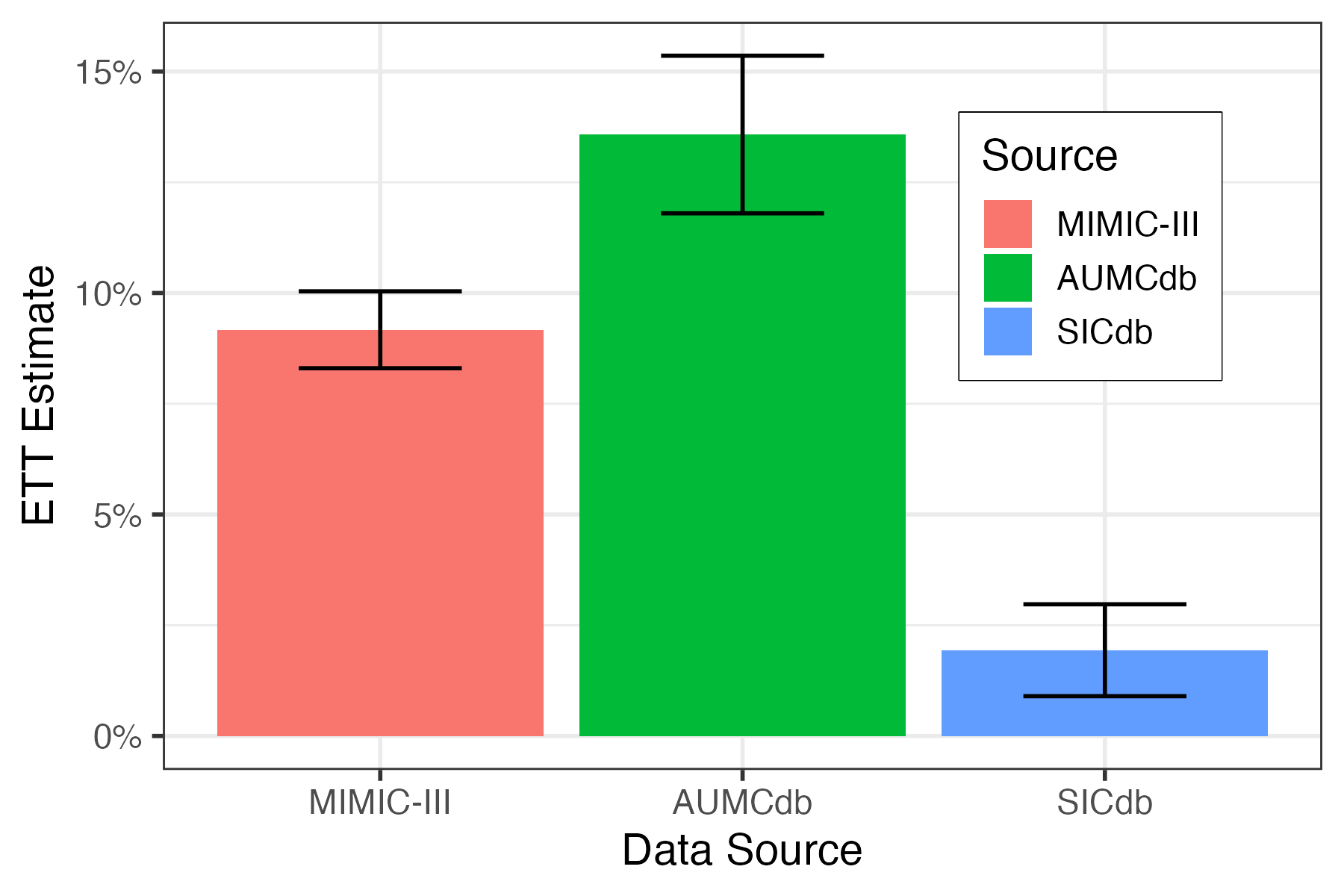}
    \caption{Estimated ETT of mechanical ventilation on in-hospital
    mortality across three ICU databases.}
    \label{fig:ett}
    \vspace{-0.1in}
\end{wrapfigure}
\begin{example}[Mechanical Ventilation from EHR data]\label{ex:mv-ett}
Mechanical ventilation $(X)$ is one of the key treatments for ICU
patients, used to support or improve oxygenation, while in-hospital mortality $(Y)$ is a typical outcome of interest. Adjusting for a set of observed covariates $Z$ (including age, sex, SOFA score, and physiological signals; see~\cref{sec:methods}), we estimate the effect of the treatment on
the treated,
\begin{align} \label{eq:ett-intro}
    \mathrm{ETT} \coloneqq \ex[Y_{x_1} - Y_{x_0} \mid X = x_1], 
\end{align}
across three large ICU databases (MIMIC-III \citep{johnson2016mimic}, AUMCdb \citep{thoral2021sharing}, SICdb \citep{rodemund2023salzburg}). \Cref{fig:ett} shows the estimated $\widehat{\mathrm{ETT}}^{bd}$ values, which are strongly above zero in every dataset. Under the assumption of no UCs, one would conclude that mechanical ventilation
increases mortality among the patients who actually received it.
While it is possible that some patients in the broader population may be harmed by this treatment, domain experts would strongly reject such a finding of harm in the treated population, who are often ventilated for a reason. The estimate is therefore a near-certain signature of UCs.
\end{example}

How does one proceed in the presence of such unobserved confounding?
The first broad alternative is to collect experimental data (for instance, for the ETT in \cref{eq:ett-intro} one could apply \emph{counterfactual randomization} \citep{bareinboim2015bandits}), which would control for UCs by design. 
However, this alternative does not come without its own difficulties: carrying out trials requires
significant resources in terms of cost, study design, preparation, patient
recruitment, and conduct, and ethical or logistical constraints may
further limit their feasibility -- or rule them out entirely
\citep{harhay2019randomized, schulz2018randomised}. 
The second alternative is to leverage domain knowledge to identify additional covariates that should be measured, and to collect data on them in a future iteration of the study. Eliciting such domain knowledge directly, however, may be non-trivial: experts are rarely trained to reason in such terms, and asking an ICU physician to enumerate the
unmeasured causes of the mechanical-ventilation-to-mortality
relationship is a daunting exercise.
Interestingly, systematic ways for eliciting unmeasured confounders have received little investigation.

In this paper, we propose an observational study design that attempts to systematically elicit candidate unobserved confounders from a decision-maker (DM) who assigns treatment. The key observation is that any $U$ that confounds the $X \to Y$ relationship must be available to the DM at the time of the decision -- otherwise it could not affect $X$. Suppose then that (i) we may interact with the DM, and (ii) the DM has the cognitive
ability to compare different sample pairs. 
Under these two assumptions, rather than asking them to list unmeasured confounders in the abstract, we may ask them to compare \emph{pairs} of units $(i,j)$ with $X_i = 1$ (treated) and $X_j = 0$ (untreated), with the intention of eliciting why the treatment decision differed.

\addtocounter{example}{-1}
\begin{example}[Continued -- Comparing Ventilated Pairs]\label{ex:pairs}
Let $Z^{(1)}$ denote the SOFA score, a standard ICU severity score
ranging from 0 to 24, with higher values corresponding to greater illness severity. Let $U^{(1)}$ be the patient's difficulty breathing perceived by the physician, which is not recorded in the EHR. We consider
three candidate pair proposals, where in each a treated patient $i$
($X_i = 1$) and an untreated patient $j$ ($X_j = 0$) are compared:
\begin{enumerate}[label=(\arabic*), leftmargin=1.75em, nosep]
    \item Patients $i$ and $j$ have the same SOFA score,
    $Z_i^{(1)} = Z_j^{(1)} = 5$, meaning neither is sicker along the observed $Z^{(1)}$. If asked for a reason for different treatment between $i$ and $j$, a physician who made the decision with access to $(Z^{(1)}, U^{(1)})$ would most likely not cite $Z^{(1)}$ as the explanation of their decision. If $U^{(1)}_i = 0$, while $U^{(1)}_j = 1$, variable $U^{(1)}$ would be a more likely explanation.
    \item The treated patient $i$ has $Z_i^{(1)} = 5$ while the untreated patient $j$ has $Z_j^{(1)} = 10$, meaning patient $j$ is sicker along $Z^{(1)}$. In such a setting, the DM would again likely not cite $Z^{(1)}$ as the explanation for their treatment decision.
    \item Finally, suppose that $Z^{(2)}$ (sex) is available, and we have
    \begin{align}
        Z^{(1)}_i &= 5,\; Z^{(2)}_i = 1,\; U^{(1)} _i = 0, \\
        Z^{(1)}_j &= 5,\; Z^{(2)}_j = 0,\; U^{(1)} _j = 1.
    \end{align}
    In such a setting, determining $U^{(1)}$ as the explanation of differing treatment would be made more difficult compared to Case~1, due to the additional existence of variable $Z^{(2)}$ along which the patients differ, which may deter the DM.
\end{enumerate}
\vspace{-0.23in}
\end{example}
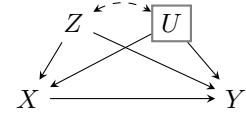
\begin{wrapfigure}{r}{0.26\textwidth}
    \centering
    \vspace{-0.2in}
    \begin{tikzpicture}
	 [>=stealth, rv/.style={thick}, rvc/.style={triangle, draw, thick, minimum size=7mm}, node distance=18mm]
	 \pgfsetarrows{latex-latex};
	 \begin{scope}
		\node[rv] (z) at (-0.65,1) {$Z$};
        \node[rv,draw=gray] (u) at (0.65,1) {$U$};
	 	\node[rv] (x) at (-1.25,0) {$X$};
	 	\node[rv] (y) at (1.5,0) {$Y$};
	 	\draw[->] (x) -- (y);
        \draw[->] (z) -- (y);
        \draw[->] (u) -- (y);
        \draw[->] (z) -- (x);
        \draw[->] (u) -- (x);
        \path[<->] (z) edge[bend left = 30, dashed] (u);
	 \end{scope}
	 \end{tikzpicture}
    \caption{Causal diagram.}
    \label{fig:dag}
    \vspace{-0.2in}
\end{wrapfigure}
\!\!The above example illustrates some of the key principles behind our framework. Case~(1) argues that if patients match on $Z$, the DM is forced to name a variable outside $Z$ to explain the treatment difference -- e.g., $U^{(1)}$. 
Case~(2) leverages a natural monotonicity consideration of this setting: if the decisions are monotone with respect to $Z^{(1)}$ (which is the case with the SOFA score in ICU data), then seeing the sicker patient untreated would render $Z^{(1)}$ an unlikely explanation -- again pointing to $U^{(1)}$.
Case~(3) shows that comparisons become harder whenever patients differ
along observed covariates that are not clearly monotonically related to the outcome (or when monotonic covariates point to different illness levels), since such differences invite explanations related to $Z$, and do not reveal UCs. 
These observations raise the central question of this paper: are there systematic, principled ways to propose patient pairs to a domain expert
so that the elicited explanations are likely to correspond to genuine
unobserved confounders? Specifically, our contributions in the paper are as follows:
\begin{enumerate}[label=(\roman*), leftmargin=*]
    \item We formalize a new observational study design, \emph{confounder
    detection via treatment intent} (CDTI). We further develop a framework
    for analyzing such a study design, comprising a matching strategy $\M$ (how sample pairs are proposed to a human annotator) and an extraction strategy $\eps$ (how the decision-maker reasons over a pair).
    \item We introduce specific matching strategies -- $Z$-matching, $\pi$-matching, and
    $Z$-dominance -- 
    and provide theoretical support for them by establishing
    stochastic-dominance results that characterize when each strategy
    yields pair distributions informative about $U$ (\cref{thm:stoch-dom}). We then establish further results on dominance between strategies (\cref{thm:zdom-beats-zmatch,thm:pi-beats-rand}).
    \item We illustrate our study design based on ICU data. First, we
    provide evidence that EHR-based estimates of the effect of common
    ICU interventions are strongly confounded (\Cref{ex:mv-ett}). 
    Then, we further validate our study design based on
    semi-synthetic MIMIC-III data.
\end{enumerate}

\paragraph{Related work.}
The ideas in this paper are related to counterfactual randomization \citep{bareinboim2015bandits}, which acknowledges that the DM has access to information about $U$ that is not captured in the recorded covariates. Our work, however, uses this observation with a different purpose, namely for eliciting UCs through structured queries. 
Further, our work is also related to research on combining observational and experimental data to sharpen causal conclusions \citep{athey2020combining, rosenman2020combining,lee2020general}; in contrast, however, our study design remains purely observational and does not modify treatment decisions; in other words, it does not require experimentation on the treatment variable $X$.
\section{Confounder Detection via Treatment Intent}\label{sec:methods}
We now describe our theoretical setting. We consider a treatment variable $X$ (e.g., mechanical ventilation in ICU), outcome variable $Y$ (in-hospital mortality), a set of observed confounders $Z$ (SOFA score, age, sex), and a set of unobserved confounders $U$ (e.g., presence of hemothorax or other lung complications) as in \cref{fig:dag}. We have data on $(Z, X, Y)$. Our goal is to recover some of the unobserved confounders $U$, where we assume access to some kind of proxy information for $U$. 
In practice, the idea is that we may be able to query the DM controlling $X$, and ask them to compare two units: unit $i$ who was treated $(X_i = 1)$ and unit $j$ who was not ($X_j = 0$). If it happens that unit $i$ is smaller than unit $j$ on each observed confounder, $Z_i^{(k)} < Z_j^{(k)}$ for all $k$, then the DM may be able to identify some of the $U$ variables: e.g., it may happen that $U^{(l)}_i > U^{(l)}_j$, so that for the unobserved confounder $U^{(l)}$ unit $i$ is larger, which actually explains the decision.
%
Theoretically, the first important part of our framework is a \emph{matching strategy} $\M$. $\M$ serves for generating \emph{proposals} that the DM can then compare; e.g., $\M$ may generate pairs $(i, j)$ that are sent to the human annotator. We consider several matching strategies: 
\begin{enumerate}[label=(\roman*), leftmargin=2em]
    \item $\Mzmatch$ picks pairs with $X_i = 1, X_j=0$ and $Z_i = Z_j$, 
    \item $\Mpimatch$ picks pairs with $X_i = 1, X_j=0$ and $\pi(Z_i) = \pi(Z_j)$ where $\pi(Z)=P(X =1 \mid Z)$,
    \item $\Mzdom$ picks pairs with $X_i = 1, X_j=0$ and $Z_i \leq Z_j$ coordinatewise, 
    \item $\Mmarg$ picks pairs $(i, j)$ with just $X_i = 1, X_j = 0$; this strategy serves as a baseline.
\end{enumerate}
We begin by immediately stating one of our main theoretical results, which anchors our framework and formalizes the intuition from \cref{sec:intro} (all proofs are provided in \cref{app:proofs}): 
\begin{theorem}[Stochastic dominance of matching strategies] 
\label{thm:stoch-dom}
Under appropriate assumptions $\mathcal{A}^{(1)}, \mathcal{A}^{(2)}, \mathcal{A}^{(3)}$, respectively, 
for each strategy $\M \in \{\Mzmatch, \Mpimatch, \Mzdom\}$, we have:
\begin{align}
    P(U \mid Z = z, X = 1) &\;\succeq_{st}\; P(U \mid Z = z, X = 0) \label{eq:zmatch-stoch} \\
    P(U \mid \pi = p, X = 1) &\;\succeq_{st}\; P(U \mid \pi = p, X = 0) \label{eq:pimatch-stoch}\\
    P(U \mid Z = z', X = 1) &\;\succeq_{st}\; P(U \mid Z = z, X = 0) \text{ if }  z' < z, \label{eq:zdom-stoch}
\end{align}
where the $\succeq_{st}$ denotes the multivariate stochastic order, $A \succeq_{st} B \implies \ex[\phi(A)] \geq \ex[\phi(B)]$ for every coordinatewise non-decreasing $\phi: \R^d \to \R$. 
\end{theorem}
Intuitively, the theorem shows that conditioning on $X = 1$ vs. $X=0$ (for different additional contexts $C=c$) results in a \emph{probabilistically larger} $U$ for the treated unit, which justifies why a matching strategy may work.
With this target established, we now build towards \cref{thm:stoch-dom} sequentially, unpacking the necessary technical conditions and providing the intuition for each matching strategy.

\subsection{\texorpdfstring{$Z$}{Z}-Matching}
We start with the $Z$-matching strategy, and first consider the basic case of univariate $U$:
\begin{proposition}[$Z$-matching via MLR]\label{prop:z-match-mlr}
Assume $U \in \R$ and that $P(X = 1 \mid Z = z, U = u)$ is non-decreasing in $u$ for every $z$. Then
\begin{align} \label{eq:uni-z-ord}
    P(U \mid Z = z, X = 1) \;\succeq_{st}\; P(U \mid Z = z, X = 0).
\end{align}
\end{proposition}
Our first result illustrates that if $U$ has a monotonic effect on the treatment $X$ for each $Z=z$, then conditioning on $X = 1$ vs. $X=0$ (for a fixed $Z=z$) results in a probabilistically larger $U$. In the context of our running example (see \cref{ex:pairs}), this means that for two patients with the same SOFA score $Z^{(1)}=5$, the one treated ($X=1$) has a probabilistically larger difficulty breathing $U^{(1)}$ compared to the untreated one ($X=0$). 
We note that the monotonic effect assumption is plausible in many health contexts, and many of the covariates used in ICU to assess patient state satisfy this property.
Thus, the above result provides the basis for why it makes sense to send $Z$-matched treated-untreated unit pairs for comparison to a human annotator.
We next move onto the case of multivariate $U$, and prove a key stochastic dominance result for it:
\begin{proposition}[$Z$-matching, multivariate $U$]\label{prop:z-match-multi}
Assume $U \in \R^d$ and:
\begin{enumerate}[label=(\roman*), leftmargin=1.75em, nosep]
    \item $P(X = 1 \mid Z = z, U = u)$ is non-decreasing in each coordinate of $u$, for every $z$; \label{cond:pi-mono}
    \item \label{cond:puz-super} $P(U \mid Z = z)$ is log-supermodular, i.e., $P(u \mid z)\, P(u' \mid z) \le P(u \wedge u' \mid z)\, P(u \vee u' \mid z)$ for all $u, u'$, where $\wedge, \vee$ denote coordinatewise minima and maxima. For twice-differentiable densities, this is equivalent to $\frac{\partial^2 \log P(u\mid z)}{\partial u^{(i)}\partial u^{(j)}} \geq 0$; Then 
\end{enumerate}
\begin{align}
    P(U \mid Z = z, X = 1) \;\succeq_{st}\; P(U \mid Z = z, X = 0).
\end{align}
\end{proposition}
The above result shows that stochastic dominance of $Z$-matching for a univariate $U$ can be extended to a multivariate $U$, provided that the conditional $P(U \mid Z = z)$ is \emph{log-supermodular}. For twice-differentiable functions, log-supermodularity is equivalent to saying that cross-component second order derivatives $\frac{\partial}{\partial u^{(i)}\partial u^{(j)}} \log P(U \mid Z = z)$ are greater than $0$. 
For instance, in the Gaussian setting, log-supermodularity requires the off-diagonal precision matrix entries to be non-positive, $(\Sigma^{-1})_{kl} \;\leq\; 0$ (see \cref{ex:gauss-logsupermod} of \cref{appendix:examples}).
In the context of our ICU example, this means that among patients with identical observed covariates $Z=z$, the unobserved severity markers $U^{(1)}, U^{(2)}, \dots$ (e.g., perceived difficulty breathing, hemodynamic instability not captured in $Z$) are \emph{positively dependent}: a patient unusually severe along one unobserved axis is, on average, also more severe along the others. This rules out scenarios in which the unobserved factors \emph{trade off} against one another at fixed $Z=z$, but accommodates the more common clinical reality that severity tends to cluster across organ systems. \cref{ex:necessity-supermod} of \cref{appendix:examples} shows why log-supermodularity in \cref{prop:z-match-multi} is necessary.

\subsection{\texorpdfstring{$\pi$}{Pi}-Matching}
We next move onto the $\pi$-matching strategy, showing that it inherits the result from the $Z$-matching case, just by invoking the propensity's covariate balancing property $X \ci Z \mid \pi(Z)$.
Specifically, we have the following key result:
\begin{proposition}[$\pi$-matching, multivariate $U$]\label{prop:pi-match-multi}
Under the monotonicity and log-supermodularity assumptions of \cref{prop:z-match-multi}, for every $p \in (0, 1)$,
\begin{align} \label{eq:pi-match-dom}
    P(U \mid \pi(Z) = p, X = 1) \;\succeq_{st}\; P(U \mid \pi(Z) = p, X = 0).
\end{align}
\end{proposition}
The covariate balancing property ensures that within a level set $\{Z : \pi(Z) = p\}$, the distribution of $Z$ is the same across treated and untreated groups; hence the only systematic difference between the two groups is the one induced by $U$, which is exactly the shift quantified by \cref{prop:z-match-multi}. 
More generally, the same argument applies to any \emph{balancing score} $b(Z)$ satisfying $X \ci Z \mid b(Z)$, of which $Z$ itself and $\pi(Z)$ are the two extremes: $Z$ being the multi-dimensional, most granular score, and $\pi$ being the one dimensional, coarsest score.
In the context of our ICU example, this result means that pairs of patients with the same propensity for mechanical ventilation -- but where one was actually ventilated and the other was not -- exhibit a systematic shift in unobserved severity in the expected direction, under the same assumptions required for $Z$-matching. This is the property that makes $\pi$-matching practically attractive: it admits exact matches in a one-dimensional score even when $Z$ is high-dimensional, while preserving the structural guarantee of $Z$-matching.

\subsection{\texorpdfstring{$Z$}{Z}-Dominance}
The previous two results exploit conditional \emph{$X$-dominance}: conditioning on $X = 1$ vs. $X = 0$ at fixed $Z = z$ (or $\pi(Z) = p$) induces a specific shift in $U$.
In practice, however, exact matching on $Z=z$ may be difficult, especially when $Z$ is high-dimensional. Suppose instead we are able to find two units with $Z_i = z'$, $Z_j = z$, where $z' < z$ coordinatewise. Intuitively, the treated unit $i$ is smaller along each observed dimension, which should make it more likely that the reason for treatment lies in the hidden variables $U$.
This intuition relies on chaining two stochastic shifts. The first is the $Z$-matching shift of \cref{prop:z-match-multi}, comparing the two treatment groups at the same $Z$. The second is a new shift, comparing two values of $Z$ at a fixed treatment level $X =x$, called \emph{$Z$-dominance}:
\begin{align}
    P(U \mid Z = z', X = x) &\;\succeq_{st}\; P(U \mid Z = z, X = x) \quad \text{for } z' < z . \label{eq:z-shift-x}
\end{align}
The reasoning behind \cref{eq:z-shift-x} is that, among (un)treated units, smaller observed $z'$ means the propensity to be (not) treated must have been compensated for along some other axis -- namely $U$. Chaining either shift with $Z$-matched $X$-dominance gives the desired result in \cref{eq:zdom-stoch}.
We now identify when \cref{eq:z-shift-x} holds, captured in the following result: 
\begin{proposition}[$Z$-dominance, multivariate $U$]\label{prop:z-dom}
Assume conditions of \cref{prop:z-match-multi} and:
\begin{enumerate}[label=(\roman*), leftmargin=1.75em, nosep]
    \addtocounter{enumi}{2}
    \item For some $x \in \{0, 1\}$, defining $h^{(x)}(z, u) \coloneqq P(X = x \mid z, u)\, P(u \mid z)$, the function $\log h^{(x)}$ is twice differentiable and strictly positive, with
    \begin{align}
        \frac{\partial^2 \log h^{(x)}}{\partial u^{(j)}\, \partial u^{(k)}}(z, u) &\;\geq\; 0 \quad \text{for all } j \neq k, \quad
        \frac{\partial^2 \log h^{(x)}}{\partial z^{(l)}\, \partial u^{(j)}}(z, u) \;\leq\; 0 \quad \text{for all } l, j. \label{eq:across-zu}
    \end{align} \label{cond:z-dom-cross}
\end{enumerate}
Then, for $z' < z$,
    $P(U \mid Z = z', X = 1) \;\succeq_{st}\; P(U \mid Z = z, X = 0)$.
\end{proposition}
\paragraph{Interpreting condition \ref{cond:z-dom-cross}: collider vs. bidirected.}
The inequalities in \ref{cond:z-dom-cross} become more transparent once $\log h^{(x)}$ is split into the contributions from the two channels through which $Z$ and $U$ interact. Writing
\begin{align}
    \log h^{(x)}(z, u) \;=\; \underbrace{\log P(X = x \mid z, u)}_{T_C \; (\text{collider channel: } Z \to X \leftarrow U)} \;+\; \underbrace{\log P(u \mid z)}_{T_B \; (\text{bidirected channel: } Z \bidir U)},
\end{align}
each cross derivative in \ref{cond:z-dom-cross} decomposes into a \emph{collider contribution} $T_C$ from the propensity term and a \emph{bidirected contribution} $T_B$ from the marginal $Z$-$U$ dependence (see \cref{fig:dag}). Under our assumptions, these two contributions carry opposite signs, so each inequality reduces to a question of which dominates.
Consider first the within-$U$ part of \cref{eq:across-zu}. Under monotonicity of $\pi$ in $u$ (assumption \ref{cond:pi-mono}), $T_C$ tends to be log-submodular in $u$, contributing a non-positive cross derivative. By assumption \ref{cond:puz-super}, $T_B$ is log-supermodular in $u$, contributing a non-negative one. The within-$U$ inequality therefore asks that $T_B$ dominates: the positive marginal dependence among the $U$ coordinates must be strong enough to survive the collider's re-weighting through $T_C$.
The across-$Z$-$U$ inequality flips both the required sign and the dominant channel. Under \ref{cond:pi-mono}, the cross derivative of $T_C$ in $z$-$u$ is typically non-positive (the collider couples $Z$ and $U$ negatively at fixed $X$); whenever $Z$ and $U$ are marginally positively dependent, the cross derivative of $T_B$ is non-negative. The across-$Z$-$U$ inequality therefore asks that $T_C$ dominates: conditioning on $\{X = x\}$ must be strong enough to flip the marginal positive $Z$-$U$ dependence into a net negative one.

In the context of our ICU example, condition \ref{cond:z-dom-cross} constrains how the propensity for mechanical ventilation interacts with the marginal dependence between observed severity (e.g., SOFA score, oxygen saturation) and unobserved severity (e.g., perceived difficulty breathing, hemodynamic instability not captured in $Z$). It does not require $Z$ and $U$ to be independent, and the two-channel decomposition above gives a concrete reading of when it is more or less plausible. A natural regime in which \ref{cond:z-dom-cross} is satisfied is one where $Z$ and $U$ correspond to different organ systems, e.g., $Z$ captures cardiovascular markers while $U$ captures respiratory markers. In such a setting, intra-system severities (within $U$, e.g., difficulty breathing and respiratory rate) tend to co-vary tightly, supporting the within-$U$ inequality, while cross-system severities (between $Z$ and $U$) are positively but more loosely related, leaving room for the collider channel to overturn the marginal $Z$-$U$ dependence. The opposite regime is also possible: if $Z$ and $U$ relate to the same organ system, the marginal $Z$-$U$ dependence may be strong enough that the collider cannot overturn it, in which case the across-$Z$-$U$ inequality fails and $Z$-dominance offers no guarantee. Whether \ref{cond:z-dom-cross} holds is therefore an empirical question about how the observed and unobserved severity markers partition across systems, and we view it as an assumption the analyst should reason about in light of the specific covariates available, rather than one that is guaranteed by the clinical setting alone. Put formally, \cref{prop:z-dom} may fail if there is either strong correlation of $Z, U$ or a strong interaction of $Z, U$ in $f_X$ (see \cref{ex:zdom-fail-interaction,ex:zdom-fail-correlation}).
Finally, we remark that the chain argument of \cref{prop:z-dom} extends to $\pi$-matching with a single, scalar version of assumption \ref{cond:z-dom-cross}, as discussed in \cref{appendix:pi-dom}.

\subsection{Extraction Strategies}\label{sec:extraction}
We now turn to the second part of our framework: how the proposed pair $(i, j)$ is processed by the expert annotator. We decompose this into three stages: extraction, selection, and success. Together, these stages determine the probability that the elicited explanation corresponds to a genuine UC.
\paragraph{Extraction.}
Given a pair $(i, j)$ with $X_i = 1, X_j = 0$, the expert produces a candidate set $C_\eps(i, j)$ of explanations for the differential treatment. We emphasize that an explanation is inherently \emph{contrastive}: a variable $V^{(l)}$ (observed or unobserved) is a candidate explanation only if $V^{(l)}_i > V^{(l)}_j$. The set $C_\eps(i, j)$ is therefore a subset of the variables along which $i$ exceeds $j$, generated according to an \emph{extraction strategy} $\eps$. 
Under \emph{perfect extraction} $\eps_\mathrm{perf}$, the extracted set correctly includes all such variables: $C_{\eps_\mathrm{perf}}(i, j) = \{V^{(l)} : V^{(l)}_i > V^{(l)}_j\}$.
In a more general setting, such as \emph{ablation-based extraction} $\eps_\mathrm{abl}$ (natural when the expert is a large language model performing counterfactual reasoning over textual notes), $C_\eps(i, j)$ may contain \emph{hallucinated} variables that did not drive the decision, or \emph{omit} genuine drivers that influenced the decision.

\paragraph{Selection.} 
Given $C_\eps(i, j)$, we assume the expert selects a single explanation uniformly at random among all candidates they consider plausible. These include unobserved candidates $U^{(l)}$, and observed coordinates $Z^{(k)}$ along which $i$ exceeds $j$. The probability that the expert selects an unobserved variable ($E = 1$) is therefore
\begin{align}
    P(E = 1 \mid i, j) = \frac{|C_\eps(i, j) \setminus Z|}{|C_\eps(i, j)|}. \label{eq:selection}
\end{align}

\paragraph{Success.} 
Selection alone is  not sufficient, as the selected variable must correspond to a genuine UC to be informative. We say the extraction is \emph{accurate} ($A = 1$) when $E = 1$ and the selected variable $c$ satisfies $U^{(c)}_i > U^{(c)}_j$.
Under perfect extraction, accuracy simplifies to
\begin{align}
    P(A = 1 \mid i, j, \eps_\mathrm{perf}) = \frac{N(i, j)}{N(i, j) + D(i, j)}, \label{eq:acc-perf}
\end{align}
with $N(i, j) = |\{U^{(l)} : U^{(l)}_i > U^{(l)}_j\}|$ counting the unobserved drivers, and $D(i, j) = |\{Z^{(k)} : Z^{(k)}_i > Z^{(k)}_j\}|$ counting the competing observed explanations. A good matching strategy should maximize $N$ (number of active UCs) and minimize $D$ (number of active known confounders).
Our first result under the rubric of \emph{strategy dominance} shows that $Z$-dominance provably increases accuracy compared to $Z$-matching under perfect extraction:
\begin{theorem}[$Z$-dominance dominates $Z$-matching]\label{thm:zdom-beats-zmatch}
Fix $Z = z$ and assume conditions \ref{cond:pi-mono}--\ref{cond:z-dom-cross} of \cref{prop:z-dom} hold at $X = 1$. Then
\begin{align} \label{eq:z-dom-above-match}
    \ex[A \mid \Mzdom, Z_i = z', Z_j = z, \eps_\mathrm{perf}] \;\geq\; \ex[A \mid \Mzmatch, Z_i = Z_j = z, \eps_\mathrm{perf}] \; \forall z' < z.
\end{align}
\end{theorem}

\subsection{\texorpdfstring{$\pi$}{Pi}-Matching Beats Marginal Matching}
\label{sec:linearized}
While \cref{thm:zdom-beats-zmatch} establishes strategy dominance for $Z$-dominance over $Z$-matching, proving a similar result for $\pi$-matching via the non-linear success probability $N / (N + D)$ requires strong assumptions over a two-dimensional $(N, -D)$ distribution. To make strategy comparisons more tractable, we introduce a linearized surrogate utility
\begin{align} \label{eq:lin-util}
    \mathcal{U}_\alpha(\M) \;\coloneqq\; \ex[N - \alpha D \mid \M], 
    \qquad \alpha > 0,
\end{align}
which retains a qualitative interpretation similar to $N/(N+D)$. For our analysis, we quantify a variable's predictive power via its marginal Area Under Curve (AUC), defined inline as $\aucsel(V^{(k)}) \coloneqq P(V^{(k)}_i > V^{(k)}_j \mid X_i = 1, X_j = 0)$ for independent draws $V_i \sim P(V \mid X = 1)$ and $V_j \sim P(V \mid X = 0)$. For the observables, $\aucsel(Z^{(k)})$ values can be computed, while for the UCs, $\aucsel(U^{(l)})$ are unknown. AUC values equal to $1/2$ correspond to independence $V^{(k)}\ci X$.

\begin{theorem}[$\pi$-matching Dominates Marginal Matching under $\mathcal{U}_\alpha$]
\label{thm:pi-beats-rand}
Assume that the conditions of \cref{prop:pi-match-multi} hold and that each marginal $Z^{(k)}$ is continuous. Then
\begin{align}
    \mathcal{U}_\alpha(\Mpimatch) - \mathcal{U}_\alpha(\Mmarg) 
    \;\geq\; 
    \alpha \sum_{k=1}^{d_Z} \big(\aucsel(Z^{(k)}) - \tfrac{1}{2}\big) 
    - \sum_{l=1}^{d_U} \big(\aucsel(U^{(l)}) - \tfrac{1}{2}\big).
    \label{eq:pi-beats-rand}
\end{align}
\end{theorem}
This theorem provides actionable insight into when $\pi$-matching outperforms marginal matching. Under \cref{prop:pi-match-multi}, $\pi$-matching guarantees a lower bound of $\ex[N \mid \Mpimatch] \geq d_U / 2$. While this expectation is generally smaller than marginal matching's $\ex[N \mid \Mmarg] = \sum_{l=1}^{d_U} \aucsel(U^{(l)})$, $\pi$-matching reduces the expected number of observed competitors strictly to $\ex[D \mid \Mpimatch] = 1/2$, compared to $\ex[D \mid \Mmarg] = \sum_{k=1}^{d_Z} \aucsel(Z^{(k)})$.
\cref{thm:pi-beats-rand} gives us that
\begin{align} \label{eq:when-pi-wins}
    \frac{\sum_{l=1}^{d_U} \big(\aucsel(U^{(l)}) - \tfrac{1}{2}\big)}{\sum_{k=1}^{d_Z} \big(\aucsel(Z^{(k)}) - \tfrac{1}{2}\big)} \leq 1 \implies \ex[N-D \mid \Mpimatch] \geq \ex[N-D \mid \Mmarg].
\end{align}
In other words, $\pi$-matching dominates when the amount of variation explained by $Z$ is large compared to the variation explained by $U$ (e.g., when $d_Z \gg d_U$, meaning most confounders are known; or when $Z^{(k)}$ AUCs are large compared to $U^{(l)}$). We verify this in the experiments.
\section{Experiments}\label{sec:experiments}
We first verify \cref{thm:zdom-beats-zmatch,thm:pi-beats-rand} on synthetic data (\cref{sec:exp-synth}), then evaluate our framework on a semi-synthetic dataset derived from MIMIC-III (\cref{sec:exp-semisynth}).

\subsection{Synthetic verification}\label{sec:exp-synth}
In both experiments, $(Z, U)$ are Gaussians $\mathcal{N}(0, \Sigma)$ with $d_Z = d_U = 3$, $Z$-block has $\Sigma_Z = I$, $U$-block $\Sigma_U = \sigma_\eps^2 I + \tau^2 \mathbf{1}\mathbf{1}^\top$. Setting $\sigma_\eps^2 = 0.4$, and $\tau^2 = 0.6$ makes $P(U \mid Z = z)$ log-supermodular, and propensity is given by $\pi(z, u) = \sigma(\alpha + \beta^\top z + \gamma^\top u)$ with $\alpha = -1$. The cross-block covariance $\Sigma_{ZU}$ is the key parameter determining whether assumptions of \cref{prop:z-dom,prop:pi-match-multi} hold true. We work under $\eps_\mathrm{perf}$ throughout and estimate expectations by Monte Carlo (\cref{fig:synthetic}).

\paragraph{\cref{thm:zdom-beats-zmatch} (Panels 1--2).} With $\beta = \gamma = \mathbf{1}$, $z = 0$, and $z' = -\delta \cdot \mathbf{1}$, we estimate $\ex[A \mid \M, \eps_\mathrm{perf}]$ for $\Mzmatch$ and $\Mzdom$ by sampling $U_i \sim P(U \mid Z, X = 1)$ and $U_j \sim P(U \mid Z = z, X = 0)$ via rejection from $\mathcal{N}(\mu(z), \Sigma_{U \mid Z})$ with acceptance $\pi^X(1-\pi)^{1-X}$. Under $\eps_\mathrm{perf}$, both strategies have $D = 0$, so $A = \mathbb{1}\{N \geq 1\}$. 
In Panel 1, $\Sigma_{ZU} = \rho I$ with $\rho = 0.05$, which means mild $Z$-$U$ dependence, and that \ref{cond:z-dom-cross} holds at $X = 1$. In Panel 2, $\rho = 0.55$, implying strong $Z$-$U$ dependence, and the across-$Z$-$U$ inequality flips at $X = 1$.
Panel~1 confirms \cref{eq:z-dom-above-match}: $\Mzdom$ dominates $\Mzmatch$ as $\delta$ grows under \cref{prop:z-dom}. Panel~2 illustrates the failure mode: when condition \ref{cond:z-dom-cross} is violated, for larger values of the dominance gap $\delta$, $Z$-matching outperforms $Z$-dominance.

\paragraph{\cref{thm:pi-beats-rand} (Panels 3--4).} 
Using the same parametric family, we fix $\beta = 0.2 \cdot \mathbf{1}$, $\gamma = 0$ and sweep $\Sigma_{ZU} = c \cdot \mathbf{1}\mathbf{1}^\top$ with $c \in [0, 0.43]$ (upper end keeps $P(U \mid Z)$ log-supermodular). Each $c$ pins down a value of the AUC ratio
$\kappa \coloneqq \sum_{l} (\aucsel(U^{(l)}) - \tfrac12) / \sum_{k} (\aucsel(Z^{(k)}) - \tfrac12)$
appearing in \cref{eq:when-pi-wins}, which we estimate from a simulated population ($n = 10^5$) using random treated/untreated pairs. We then collect $\Mpimatch$ pairs by rejection on $|\pi(Z_i) - \pi(Z_j)| < 0.5\%$, with $\pi(Z)$ computed by marginalizing out $U$ in the $\sigma$ function. We plot $\ex[N - D \mid \M]$ for both $\Mpimatch$ and $\Mmarg$ as a function of $\kappa$. 
Parameter $\gamma = 0$ is chosen so \cref{eq:pi-beats-rand} holds with equality and the threshold $\kappa = 1$ is sharp, meaning that $\ex[N - D \mid \Mpimatch] - \ex[N - D \mid \Mmarg] = (1 - \kappa) \sum_k (\aucsel(Z^{(k)}) - \tfrac12)$, so the linearized utility gap $\U_\alpha(\Mpimatch) - \U_\alpha(\Mmarg)$ from \cref{eq:lin-util} changes sign exactly at $\kappa = 1$. 
Panel~3 for $\kappa \leq 1$ confirms this, with $\Mpimatch$ above $\Mmarg$, as predicted. Panel~4 for $\kappa > 1$ shows that the inequality indeed reverses, as predicted by the theoretical result.

\begin{figure}
    \centering
    \includegraphics[width=\linewidth]{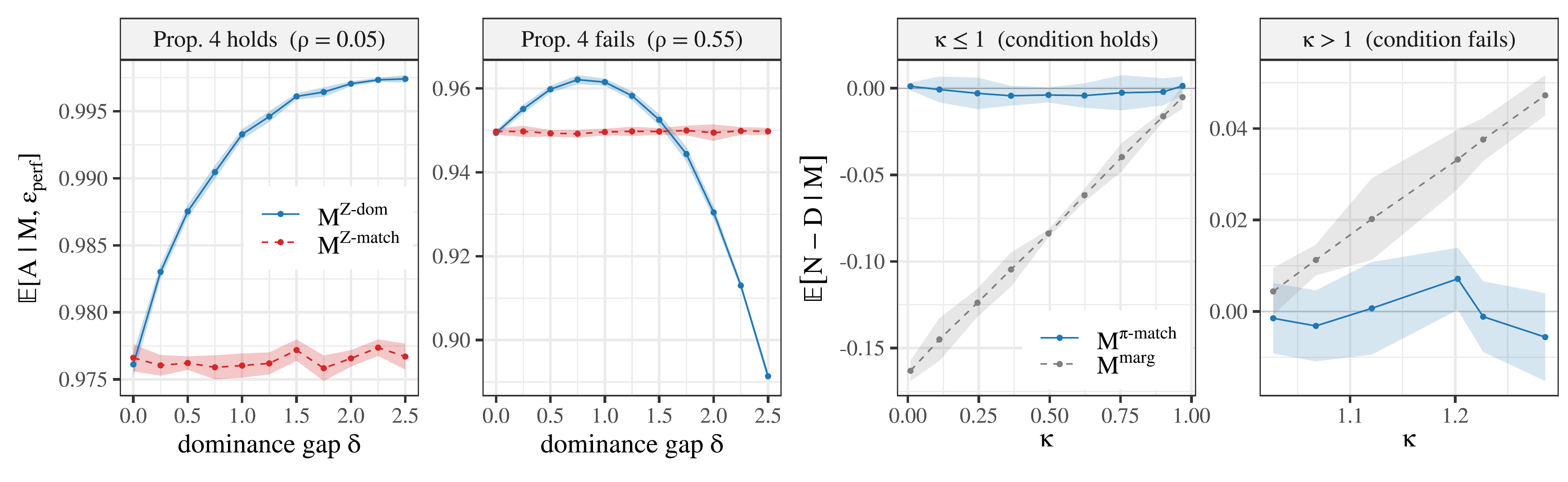}
    \caption{Synthetic verification. \emph{Panels 1--2}: $\Mzdom$ vs.\ $\Mzmatch$ as a function of the dominance gap $\delta$, with \ref{cond:z-dom-cross} holding ($\rho = 0.05$) and failing ($\rho = 0.55$). \emph{Panels 3--4}: $\Mpimatch$ vs.\ $\Mmarg$ as a function of $\kappa$, in the regimes $\kappa \leq 1$ and $\kappa > 1$. Bands are 95\% across-rep CIs.}
    \label{fig:synthetic}
\end{figure}

\subsection{Semi-synthetic experiments on MIMIC-III}\label{sec:exp-semisynth}

\paragraph{Data construction.} We build MIMIC-III-SeS by using 12 observed covariates $Z$ (resp.\ rate, MAP, lactate, P/F ratio, $p\mathrm{CO}_2$, $p\mathrm{O}_2$, $\mathrm{O}_2$ saturation, age, sex, Charlson, SOFA, plus admission-diagnosis indicators). The unobserved confounders $U$ dimension is $d_U=10$, and all are binary, extracted from clinical notes via UMLS entity linking with negation detection \citep{neumann2019scispacy}. The concepts include pleural effusion, heart failure, dyspnea, pneumonia, pulmonary edema, Bloom syndrome, hypoxia, atrial fibrillation, hypertensive disease, and atelectasis. 
We fit a logistic propensity $\pi(Z, U) = \sigma(\alpha_X + \beta_X^\top Z + \gamma_X^\top U)$ with $\gamma_X^{(l)} \sim \mathrm{Unif}[0.3, 0.4]$ (so $\pi$ is monotone in $u$) and a logistic outcome model with coefficient $\gamma_Y^{(l)} \sim \mathrm{Unif}[0.1, 0.2]$ for $U$ and $\gamma_{X\to Y}=-0.2$ for $X$, implying treatment $X$ is protective. 
We sample new $(\tilde X, \tilde Y)$ from this model, making data $(Z, \tilde X, \tilde Y)$ semi-synthetic. $U$ itself is never exposed to the matching strategies.

\paragraph{Strategy implementation.} For a practical implementation of $\Mzmatch$, we compute the Euclidean distance on standardized $Z$ values for every pair $(i, j)$ with $X_i = 1, X_j = 0$, sort in ascending order, and take the top $B$ pairs, selecting pairs closest in $Z$-space. 
$\Mzdom$ counts the coordinates where $i$ exceeds $j$ (margin $0.2\sigma$), and picks pairs where this score is smallest.
$\Mpimatch$ considers ascending $|\hat\pi(Z_i) - \hat\pi(Z_j)|$ values, with $\hat\pi$ estimated by 5-fold cross-validated logistic regression on $Z$;
$\Mpidom$ (see \cref{appendix:pi-dom}) considers ascending $\hat\pi(Z_i) - \hat\pi(Z_j)$. $\Mmarg$ considers a random order over the data.
We vary the budget $B$, which represents the number of pairs proposed by a strategy. For each selected pair we compute the probability of success $\lambda(i,j) \coloneqq P(A = 1 \mid i, j, \eps_\mathrm{perf})$ as in \cref{eq:acc-perf}. The cumulative average of $\lambda$ traces how successful the strategies are with increasing budget sizes. 

\paragraph{Results.} \Cref{fig:semisynth} (left) shows the cumulative success rate vs.\ budget $B$. The $Z$-based strategies $\Mzmatch$ and $\Mzdom$ achieve the highest success, followed by $\Mpimatch$ and $\Mpidom$, while $\Mmarg$ has lowest success throughout. 
These results empirically verify \cref{thm:zdom-beats-zmatch,thm:pi-beats-rand}. 
Further, \Cref{fig:semisynth} (right) reports mean $\lambda(i, j)$ averaged within strata of $|\hat\pi(Z_i) - \hat\pi(Z_j)|$ (low $\to$ high), and shows that $\bar\lambda$ is highest when the propensity gap is small and decreases monotonically as the gap widens (the red, dashed horizontal line shows the overall $\lambda$ average in the dataset, corresponding to the success of $\Mmarg$). This figure empirically verifies \cref{thm:pi-beats-rand}: pairs closely matched in $\hat\pi$ carry higher success probability than randomly drawn pairs, justifying $\Mpimatch$ over $\Mmarg$. The above results show that the theoretical implications developed in the manuscript hold even when the $(Z,U)$ structure is inherited from real-world clinical data.

\begin{figure}
    \centering
    \includegraphics[width=\linewidth]{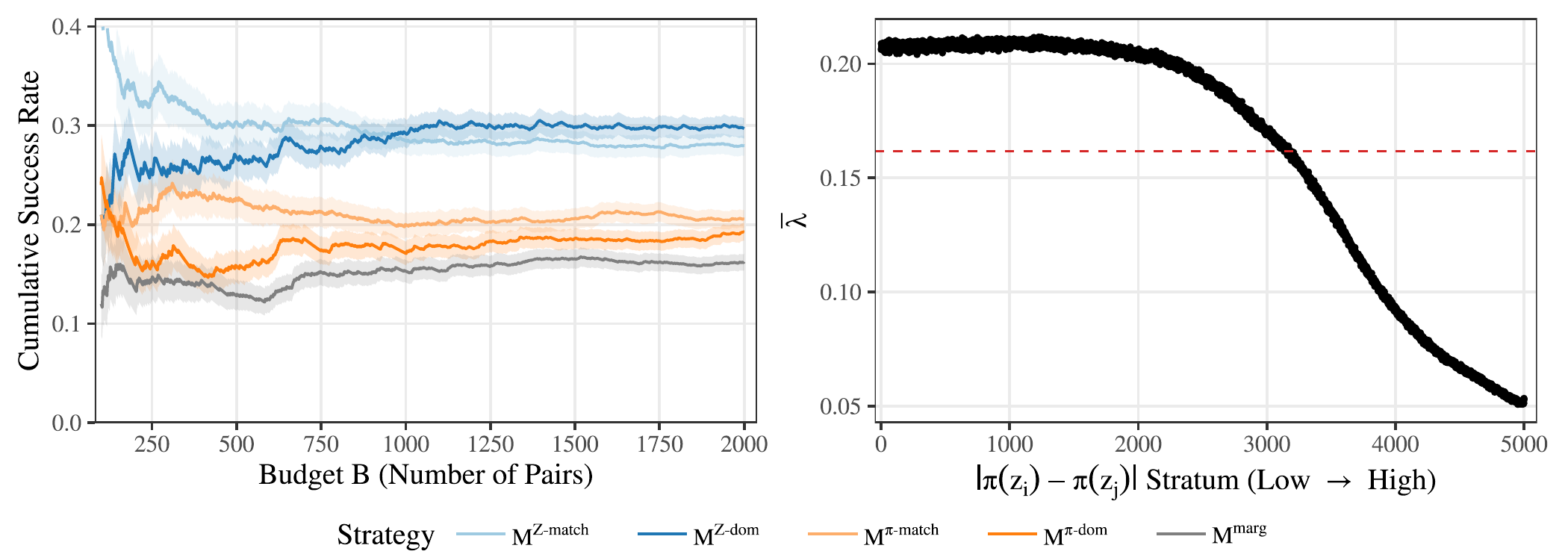}
    \caption{Semi-synthetic results on MIMIC-III-SeS. \emph{Left}: cumulative success rate vs.\ budget $B$ under $\eps_\mathrm{perf}$. \emph{Right}: mean $\bar\lambda$ per stratum of $|\pi(z_i) - \pi(z_j)|$ (low $\to$ high) on the $\Mpimatch$-ranked pool of unit pairs; dashed red is the overall mean.}
    \label{fig:semisynth}
\end{figure}

\subsection{Real-Data: Mechanical Ventilation in MIMIC-III}\label{sec:real}
We now apply our framework to the ICU setting that motivated the paper (Ex.~\ref{ex:mv-ett}). Recall that the EHR-based ETT estimate of mechanical ventilation on in-hospital mortality is strongly above zero in every database (Fig.~\ref{fig:ett}), a near-certain signature of unobserved confounding. We perform confounder detection to elicit confounder candidates for MIMIC-III, then re-estimate the ETT.

We first fit a BERT-based estimator $\hat\pi_{\mathrm{BERT}}$ that predicts $P(X=1 \mid Z, T)$ (mechanical ventilation) from the physiological covariates $Z$ together with the clinical notes $T$, using the training subset of the data.
We further fit a predictor $\hat\pi_{\mathrm{xgb}}(Z)$ for predicting $X$ just based on the covariates $Z$ using \texttt{xgboost} \citep{chen2016xgboost}, and such predicted probabilities are used for the $\Mpimatch$ strategy. 
On a held-out set, we select $B = 2{,}000$ treated-control pairs $(i, j)$ using $\Mpimatch$, ensuring that any unit $i$ or $j$ can appear up to three times. 
For each selected pair, we run ablation-based extraction $\eps_\mathrm{abl}$. We extract the UMLS concepts present in patient $i$'s notes but not in $j$'s (set denoted by $\mathcal{C}_i \setminus \mathcal{C}_j$), and for each concept $c \in \mathcal{C}_i \setminus \mathcal{C}_j$ we remove its mentions from patient $i$'s notes and recompute the propensity for treatment labeled $\hat\pi(z_i, t_i^{-c})$. We define the concept impact as $\Delta_c = \hat\pi_{\mathrm{BERT}}(z_i, t_i) - \hat\pi_{\mathrm{BERT}}(z_i, t_i^{-c})$. Concepts with $\Delta_c > 1\%$ are recorded as candidate confounders.

Aggregating across the $2{,}000$ pairs, we report the 20 most frequently discovered concepts in \cref{tab:confounders}. 
We classify the concepts into groups, based on their clinical interpretation.
Group \textbf{(A)} consists of concepts with direct indication for mechanical ventilation, related to pulmonary injury or impairment. Group \textbf{(B)} consists of concepts (fever, tachycardia, communicable disease) related to Systemic Inflammatory Response Syndrome (SIRS \citep{bone1992definitions}), which is a known pathway that drives ventilation via septic respiratory failure \citep{esteban2000mechanical}. 
Group \textbf{(C)} is related to hemorrhage and trauma, which drive MV via shock and airway protection. 
Group \textbf{(D)} contains cardiac comorbidities, which act indirectly via hemodynamic instability or background risk \citep{esteban2000mechanical}. 
Finally, Group \textbf{(E)} contains non-specific correlates of illness severity, which are less likely to be recognized as drivers of ventilation decisions. 
In summary, Group (A) has very high clinical plausibility in terms of confounding MV, while Groups (B)-(D) have high plausibility. Therefore, 17/20 detected confounders are clinically highly plausible.

\paragraph{Effect on ETT estimation.} 
We next assess whether adjusting for the detected confounders impacts the ETT estimate comparing to using $Z$ only. 
We compute the ETT on the entire MIMIC-III data using causal forests \citep{athey2019estimating}  with four adjustment sets: $Z$, $(Z, U_A)$, $(Z, U_{A,B})$, and $(Z, U_{A,B,C,D})$, with $U_G$ encoding binary indicators for the concepts in groups $G$. 
\Cref{fig:ett-real} shows the effect estimates for different adjustment sets. 
The figure indicates that incorporating $U_A, U_B$ reduces the ETT slightly, while incorporating $U_{A,B,C,D}$ pulls the estimate back to a larger value. 
However, none of the effect estimates differ to a statistically significant level. Therefore, while most of the elicited confounders are clinically plausible, accounting for their impact does not reduce the ETT towards 0. 
We hypothesize that this is due to the retrospective and partial nature of the text data: even highly relevant concepts such as pleural effusion are documented in the notes for only a minority of patients, so the binary indicators we construct may be noisy proxies rather than the underlying confounder itself.

\begin{figure}[t]
\centering
\begin{minipage}[t]{0.54\linewidth}
\centering
\small
\setlength{\tabcolsep}{4pt}
\scalebox{0.9}{
\begin{tabular}{@{}rlc@{\hspace{1.2em}}rlc@{}}
\toprule
CUI C0\# & Concept & G & CUI C0\# & Concept & G \\
\midrule
32227 & Pleural effusion     & A & 10054 & Coronary athero.    & D \\
19080 & Hemorrhage           & C & 04144 & Atelectasis         & A \\
32285 & Pneumonia            & A & 16658 & Fracture            & C \\
13604 & Edema                & E & 09450 & Commun. disease    & B \\
13404 & Dyspnea              & A & 332448 & Infiltration        & A \\
30193 & Pain                 & E & 27497 & Nausea              & E \\
151699 & Intracranial hem.   & C & 242184 & Hypoxia             & A \\
20538 & Hypertensive dis.    & D & 31039 & Pericardial eff.    & D \\
15967 & Fever                & B & 39239 & Sinus tachycardia   & B \\
34063 & Pulm. edema      & A & 32326 & Pneumothorax        & A \\
\bottomrule
\end{tabular}
}
\captionof{table}{Top 20 discovered confounders (ranked by frequency) on real data via $\eps_\mathrm{abl}$ with $B = 2{,}000$ pairs via $\Mpimatch$. Groups: (\textbf{A}) pulmonary impairment/injury; (\textbf{B}) infection/SIRS; (\textbf{C}) hemorrhage/trauma; (\textbf{D}) cardiac; (\textbf{E}) non-specific.}
\label{tab:confounders}
\end{minipage}%
\hfill
\begin{minipage}[t]{0.4\linewidth}
\centering
\vspace{-55pt}
\includegraphics[width=1.03\linewidth]{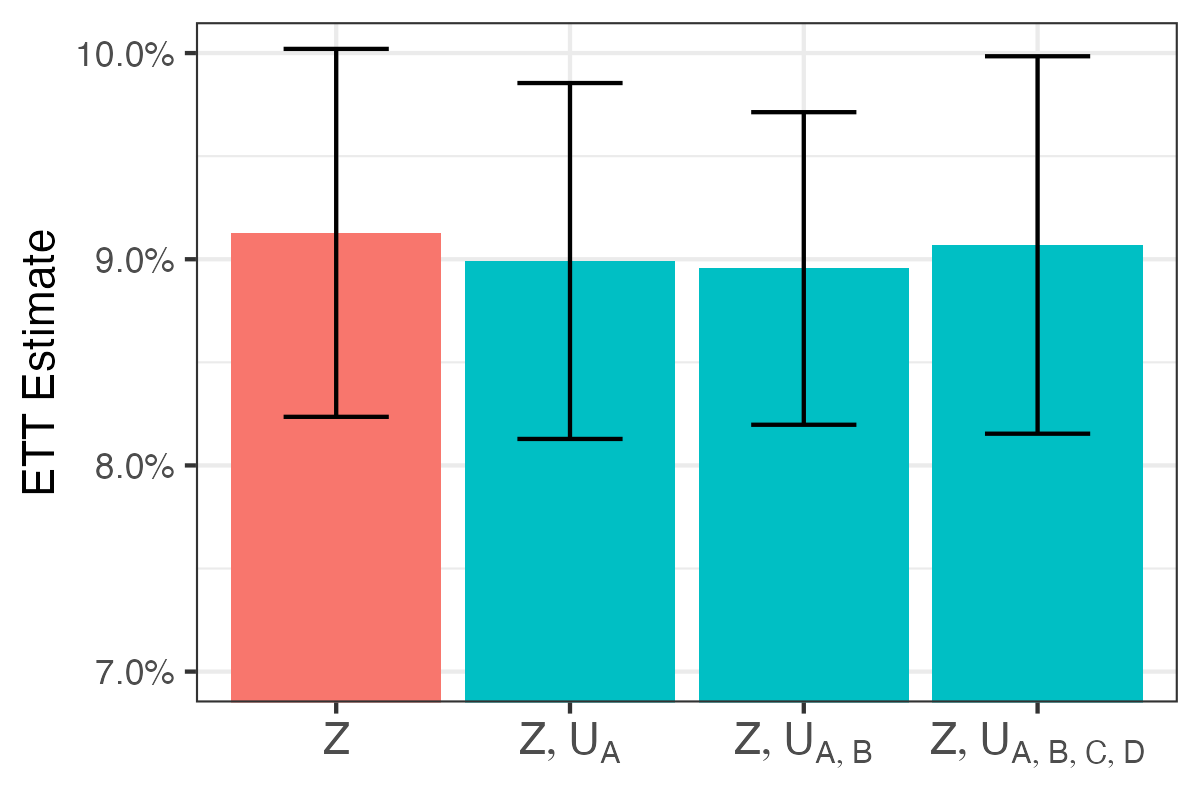}
\vspace{-13pt}
\captionof{figure}{ETT estimates on MIMIC-III with adjustment sets $Z$, $\{Z, U_A\}$, $\{Z, U_{A,B}\}$, and $\{Z, U_{A,B,C,D}\}$, using causal forests with $100$ bootstrap replicates ($95\%$ CIs reported).}
\label{fig:ett-real}
\end{minipage}
\end{figure}

\paragraph{Limitations and future work.} Several aspects of our framework merit further development. 
First, the multivariate stochastic-dominance notion underlying \cref{thm:stoch-dom} is strong; relaxing this notion to coordinate-wise marginal stochastic dominance is a natural next step that would likely allow relaxing the required assumptions (such as log-supermodularity of $P(u\mid z)$). 
Second, our analysis of extraction (\cref{sec:extraction}) rests on a specific model of expert behavior with uniform selection (\cref{eq:acc-perf}); this is a first formalization, and we expect real-world elicitation may work differently. 
Third, the framework presupposes that the decision-maker can articulate the factors driving their decisions when comparing two units; in settings where treatment is allocated on tacit or subconscious cues, the elicited explanations may underrepresent the true confounders. 
Finally, our framework detects \emph{candidates} for unobserved confounders with a probabilistic guarantee, but does not certify whether any specific detected variable is a confounder of $X \to Y$, nor does it establish that adding the detected variables recovers back-door validity. We leave these interesting directions for future work.


\newpage
\ifnum\treportflag=0
\section*{Broader Impact Statement}
This paper proposes an observational study design for eliciting candidate unobserved confounders from domain experts. The intended impact is methodological: providing causal analysts with a structured alternative to either assuming away unobserved confounding or commissioning new experimental data, which is often infeasible.

Used carefully, our framework could help analysts identify previously unmeasured covariates that ought to be collected in future iterations of a study, improving the validity of subsequent analyses. We therefore do not identify specific risks posed by the work, but view it as possibly strengthening the methodology in observational causal inference.
\fi

\bibliographystyle{abbrvnat}
\bibliography{refs}


\newpage
\appendix
\section*{\centering\Large Supplementary Material:\\ \textit{\ttl{}}}
The source code for reproducing the results can be found in the anonymized code repository \url{https://anonymous.4open.science/r/icu-deconfounding-5DB9}. The repository includes a README file explaining the setup steps. Most experiments were run on a MacBook Pro M3 (Tahoe 26.2) with under 24 hours of total computation.

\section{Proofs} \label{app:proofs}
In this appendix, we provide formal proofs of the results appearing in the main text.

\begin{proof}
Fix $z$ and let $\pi(z, u) = P(X = 1 \mid Z = z, U = u)$. By Bayes' rule,
\begin{align}
    P(u \mid Z = z, X = 1) &\propto \pi(z, u) \, P(u \mid Z = z), \\
    P(u \mid Z = z, X = 0) &\propto (1 - \pi(z, u)) \, P(u \mid Z = z).
\end{align}
The likelihood ratio is therefore
\begin{align}
    r(u) \coloneqq \frac{P(u \mid Z = z, X = 1)}{P(u \mid Z = z, X = 0)} \propto \frac{\pi(z, u)}{1 - \pi(z, u)},
\end{align}
which is non-decreasing in $u$ since $\pi(z, u)$ is. Hence the family $\{P(\cdot \mid Z = z, X = x)\}_{x \in \{0, 1\}}$ has monotone likelihood ratio in $u$, which implies the stochastic ordering in \cref{eq:uni-z-ord}.
\end{proof}

We first state a key lemma of Karlin and Rinott \citep{karlin1980classes} required for the proof of \cref{prop:z-match-multi}:
\begin{lemma}[Karlin-Rinott, 1980 \citep{karlin1980classes}]\label{lem:kr}
Let $f_0, f_1$ be densities on $\R^d$ satisfying
\begin{align}
    f_0(u)\, f_1(u') \;\le\; f_0(u \wedge u')\, f_1(u \vee u') \quad \text{for all } u, u' \in \R^d,
\end{align}
Then $f_1 \succeq_{st} f_0$.
\end{lemma}
\begin{proof}[Proof of \cref{prop:z-match-multi}]
Fix $z$ and write $\pi(u) = P(X = 1 \mid Z = z, U = u)$, $g(u) = P(u \mid Z = z)$. By Bayes' rule,
\begin{align}
    P(u \mid Z = z, X = 1) \propto \pi(u)\, g(u), \qquad P(u \mid Z = z, X = 0) \propto (1 - \pi(u))\, g(u).
\end{align}
We verify the Karlin-Rinott condition of \cref{lem:kr}. By assumption (i), $\pi(u \vee u') \geq \pi(u')$ and $1 - \pi(u \wedge u') \geq 1 - \pi(u)$, so
\begin{align}
    (1 - \pi(u))\, \pi(u') \;\le\; (1 - \pi(u \wedge u'))\, \pi(u \vee u'). \label{eq:pi-ineq}
\end{align}
Combining \cref{eq:pi-ineq} with assumption (ii),
\begin{align}
    (1 - \pi(u))\, g(u) \cdot \pi(u')\, g(u') \;\le\; (1 - \pi(u \wedge u'))\, g(u \wedge u') \cdot \pi(u \vee u')\, g(u \vee u'),
\end{align}
which is exactly the condition of \cref{lem:kr} applied to $f_0(u) \propto (1 - \pi(u)) g(u)$ and $f_1(u) \propto \pi(u) g(u)$, and the conclusion follows.
\end{proof}

\begin{proof}[Proof of \cref{prop:pi-match-multi}]
By the covariate balancing property of the propensity score, $X \ci Z \mid \pi(Z)$, which implies $P(Z \mid \pi(Z) = p, X = x) = P(Z \mid \pi(Z) = p)$ for $x \in \{0, 1\}$. Expanding over $Z$ gives:
\begin{align}
    P(U \mid \pi(Z) = p, X = x) = \int P(U \mid Z = z, X = x)\, P(z \mid \pi(Z) = p)\, dz.
\end{align}
Both treatment groups are mixtures with respect to the same mixing distribution $P(z \mid \pi(Z) = p)$. By \cref{prop:z-match-multi}, $P(U \mid Z = z, X = 1) \succeq_{st} P(U \mid Z = z, X = 0)$ for every $z$. Stochastic dominance is preserved under mixing with a common mixing measure, since for every coordinatewise non-decreasing $\phi$,
\begin{align}
    \ex[\phi(U) \mid \pi(Z) = p, X = 1] &= \int \ex[\phi(U) \mid Z = z, X = 1]\, P(z \mid \pi(Z) = p)\, dz \\
    &\geq \int \ex[\phi(U) \mid Z = z, X = 0]\, P(z \mid \pi(Z) = p)\, dz \\
    &= \ex[\phi(U) \mid \pi(Z) = p, X = 0],
\end{align}
which implies \cref{eq:pi-match-dom}.
\end{proof}

\begin{proof}[Proof of \cref{prop:z-dom}]
The target claim is decomposed into two-step chains:
\begin{align}
    P(U \mid z', X{=}1) \;\succeq_{st}\; P(U \mid z', X{=}0) \;\succeq_{st}\; P(U \mid z, X{=}0), \label{eq:x0-chain} \\
    P(U \mid z', X{=}1) \;\succeq_{st}\; P(U \mid z, X{=}1) \;\succeq_{st}\; P(U \mid z, X{=}0). \label{eq:x1-chain}
\end{align}
In each chain, one link follows from \cref{prop:z-match-multi} under (i) and (ii) (applied at $Z = z'$ in \cref{eq:x0-chain}, at $Z = z$ in \cref{eq:x1-chain}). The other link is the $Z$-shift \cref{eq:z-shift-x}, which we now establish via assumption. We treat the $x = 0$ case; $x = 1$ is symmetric with $\pi$ in place of $1 - \pi$.

Set $f_0(u) \coloneqq P(u \mid z, X = 0)$ and $f_1(u) \coloneqq P(u \mid z', X = 0)$. Both are proportional to $h^{(0)}(\cdot, u)$ at the corresponding $z$ argument, with $u$-independent normalizers. We want to show that, for all $u, u' \in \R^d$,
\begin{align}
    f_0(u)\, f_1(u') \;\leq\; f_0(u \wedge u')\, f_1(u \vee u'); \label{eq:kr-cond}
\end{align}
equivalently, after canceling normalizers,
\begin{align}
    h^{(0)}(z, u)\, h^{(0)}(z', u') \;\leq\; h^{(0)}(z, u \wedge u')\, h^{(0)}(z', u \vee u'). \label{eq:kr-h}
\end{align}
For $u \leq u'$ or $u \geq u'$ coordinatewise, \cref{eq:kr-h} is trivial. For non-comparable $u, u'$, the derivative conditions in \cref{eq:across-zu} at $X=0$ imply, by integration along an axis-aligned path from $(z', u \vee u', z, u \wedge u')$ to $(z, u, z', u')$, the four-point inequality
\begin{align}
    \log h^{(0)}(z, u) + \log h^{(0)}(z', u') \;\leq\; \log h^{(0)}(z, u \wedge u') + \log h^{(0)}(z', u \vee u'),
\end{align}
which is exactly \cref{eq:kr-h}. Applying \cref{lem:kr} to \cref{eq:kr-cond} yields $f_1 \succeq_{st} f_0$, which is the $X = 0$ link of \cref{eq:x0-chain}, completing the chain.
\end{proof}

\begin{proof}[Proof of \cref{thm:zdom-beats-zmatch}]
Note that for both $\Mzmatch$ and $\Mzdom$, we have $|\{Z^{(k)} : Z^{(k)}_i > Z^{(k)}_j\}| =0$, i.e., $D(i, j) = 0$.
Therefore, $A = \mathbf{1}\{N \geq 1\}$ for both strategies, and the comparison reduces to $P(N \geq 1)$ (a single unobserved confounder greater for $i$ than $j$ is sufficient for success).
Conditional on $(Z_i, Z_j, X_i, X_j)$, the variables $U_i$ and $U_j$ are independent with $U_i \sim P(\cdot \mid Z_i, X_i = 1)$ and $U_j \sim P(\cdot \mid Z_j, X_j = 0)$. Define
\begin{align}
    \psi(u; U_j) \coloneqq \mathbf{1}\{\exists\, l : u^{(l)} > U_j^{(l)}\},
\end{align}
which is a coordinatewise non-decreasing function in $u$ for every fixed value of $U_j$. Hence $u \mapsto \ex_{U_j}[\psi(u; U_j)]$ is also coordinatewise non-decreasing. 
For the untreated unit, both strategies condition on $Z_j = z, X_j = 0$ and thus $P(U_j \mid Z_j = z, X_j=0)$ is the same between the strategies. 
However, the distribution of $U_i$ changes from $P(\cdot \mid z, X = 1)$ under $\Mzmatch$ to $P(\cdot \mid z', X = 1)$ under $\Mzdom$. By the $X = 1$ chain link (\cref{eq:x1-chain}) of \cref{prop:z-dom}, which uses condition \ref{cond:z-dom-cross} at $X = 1$,
\begin{align}
    P(U \mid Z = z', X = 1) \;\succeq_{st}\; P(U \mid Z = z, X = 1).
\end{align}
Multivariate stochastic dominance applied to the coordinatewise non-decreasing function $u \mapsto \ex_{U_j}[\psi(u; U_j)]$, together with noting that $\ex[\psi \mid C] = P(N\geq 1 \mid C)$ conditional on any $C$, gives the claim in \cref{eq:z-dom-above-match}.
\end{proof}

\begin{proof}[Proof of \cref{thm:pi-beats-rand}]
Decompose into $N$ and $D$ contributions:
\begin{align}
    \mathcal{U}_\alpha(\Mpimatch) - 
    \mathcal{U}_\alpha(\Mmarg) 
    \;&=\; 
    \underbrace{\big(\ex[N \mid \Mpimatch] 
        - \ex[N \mid \Mmarg]\big)}_{=:\, \Delta_N} \\
    \;&\quad -\; \alpha \underbrace{\big(\ex[D \mid \Mpimatch] 
        - \ex[D \mid \Mmarg]\big)}_{=:\, \Delta_D}.
\end{align}
By linearity, $\ex[N \mid \M] = \sum_l P(U_i^{(l)} > U_j^{(l)} \mid \M)$ 
and $\ex[D \mid \M] = \sum_k P(Z_i^{(k)} > Z_j^{(k)} \mid \M)$.

\medskip
\noindent\emph{Marginal matching.} Under $\Mmarg$, 
$(Z_i, U_i) \sim P(\cdot \mid X = 1)$ and $(Z_j, U_j) \sim P(\cdot \mid X = 0)$ 
are independent. The per-coordinate probabilities therefore equal the 
corresponding marginal selection AUCs by definition:
\begin{align}
    \ex[D \mid \Mmarg] 
        = \sum_{k=1}^{d_Z} \aucsel(Z^{(k)}), 
    \qquad
    \ex[N \mid \Mmarg] = \sum_{l=1}^{d_U} \aucsel(U^{(l)}).
    \label{eq:rand-vals}
\end{align}

\noindent\emph{$\pi$-match, $D$-side.} Conditional on 
$\pi(Z_i) = \pi(Z_j) = p$, the covariate-balancing property 
$X \ci Z \mid \pi(Z)$ implies $Z_i$ and $Z_j$ are i.i.d.\ draws from 
$P(Z \mid \pi(Z) = p)$. Continuity of the marginal of $Z^{(k)}$ and 
exchangeability give
\begin{align}
    P(Z_i^{(k)} > Z_j^{(k)} \mid \pi = p) \;=\; \tfrac{1}{2}
\end{align}
for every $k$ and $p$. Summing over $k$ and averaging over the pair 
distribution of $\pi$ yields 
$\ex[D \mid \Mpimatch] = d_Z / 2$. Hence
\begin{align}
    \Delta_D \;=\; \tfrac{d_Z}{2} - \sum_k \aucsel(Z^{(k)}) 
    \;=\; -\sum_k \big(\aucsel(Z^{(k)}) - \tfrac{1}{2}\big).
    \label{eq:Delta_D}
\end{align}

\noindent\emph{$\pi$-match, $N$-side.} Conditional on $\pi = p$, 
\cref{prop:pi-match-multi} gives the multivariate ordering 
$P(U \mid \pi = p, X = 1) \succeq_{st} P(U \mid \pi = p, X = 0)$, 
which implies the marginal ordering on each coordinate:
\begin{align}
    P(U^{(l)} \mid \pi = p, X = 1) 
    \;\succeq_{st}\; 
    P(U^{(l)} \mid \pi = p, X = 0).
    \label{eq:marg-dom}
\end{align}
For two independent univariate continuous random variables $A, B$ 
with $A \succeq_{st} B$ and CDFs $F_A \leq F_B$,
\begin{align}
    P(A > B) \;=\; \ex[1 - F_A(B)] \;\geq\; \ex[1 - F_B(B)] 
    \;=\; \tfrac{1}{2},
\end{align}
where the last equality uses the probability integral transform 
$F_B(B) \sim \mathrm{Unif}[0, 1]$. Applied to \cref{eq:marg-dom} 
conditionally on $\pi = p$, this yields 
$P(U_i^{(l)} > U_j^{(l)} \mid \pi = p) \geq 1/2$. Summing over $l$ 
and averaging over $p$,
\begin{align}
    \ex[N \mid \Mpimatch] \;\geq\; d_U / 2,
\end{align}
and therefore
\begin{align}
    \Delta_N \;\geq\; \tfrac{d_U}{2} - \sum_l \aucsel(U^{(l)}) 
    \;=\; -\sum_l \big(\aucsel(U^{(l)}) - \tfrac{1}{2}\big).
    \label{eq:Delta_N}
\end{align}

\noindent\emph{Combining.} Substituting \cref{eq:Delta_D,eq:Delta_N},
\begin{align}
    \mathcal{U}_\alpha(\Mpimatch) - 
    \mathcal{U}_\alpha(\Mmarg) 
    \;&=\; \Delta_N - \alpha \Delta_D \\
    \;&\geq\;
    -\sum_l \big(\aucsel(U^{(l)}) - \tfrac{1}{2}\big) 
    + \alpha \sum_k \big(\aucsel(Z^{(k)}) - \tfrac{1}{2}\big),
\end{align}
which gives \cref{eq:pi-beats-rand}.
\end{proof}
\newpage

\section{Examples Supporting the Theory} \label{appendix:examples}
\ifexamples
The following example illustrates log-supermodularity in the case of a Gaussian distribution:
\begin{example}[Log-supermodularity for Gaussians] \label{ex:gauss-logsupermod}
Suppose $U \mid Z = z \sim \mathcal{N}(\mu(z), \Sigma)$, where the mean $\mu(z)$ may depend on $z$ but the covariance $\Sigma$ does not. The log-density is
\begin{align}
    \log P(u \mid z) = -\tfrac{1}{2}(u - \mu(z))^\top \Sigma^{-1} (u - \mu(z)) + \mathrm{const},
\end{align}
and a direct computation gives, for $k \neq l$,
\begin{align}
    \frac{\partial^2 \log P(u \mid z)}{\partial u^{(k)} \partial u^{(l)}} = -(\Sigma^{-1})_{kl}.
\end{align}
Log-supermodularity therefore reduces to the condition
\begin{align}
    (\Sigma^{-1})_{kl} \;\leq\; 0 \quad \text{for all } k \neq l,
\end{align}
i.e., the off-diagonal entries of the precision matrix are non-positive, which is equivalent to all partial correlations being non-negative.
\end{example}
The next example illustrates why log-supermodularity is necessary in \cref{prop:z-match-multi}:
\begin{example}[Necessity of supermodularity]\label{ex:necessity-supermod}
Fix $Z=z$ and condition on it throughout. Consider $U = (U^{(1)}, U^{(2)}) \in \{0, 1\}^2$, and the joint distribution
\begin{align}
    P(U) = \begin{cases} 0.49 & \text{if } u \in \{(1, 0), (0, 1)\} \\ 0.01 & \text{if } u \in \{(0, 0), (1, 1)\}, \end{cases}
\end{align}
which is log-submodular, since $P(0,0)\, P(1,1) = 10^{-4} \ll 0.2401 = P(0, 1)\, P(1, 0)$. Let the propensity depend only on $U^{(1)}$,
\begin{align}
    \pi(u^{(1)}, u^{(2)}) = \begin{cases} 0.9 & \text{if } u^{(1)} = 1 \\ 0.1 & \text{if } u^{(1)} = 0, \end{cases}
\end{align}
which is non-decreasing in each coordinate of $u$. A direct computation gives
\begin{align}
    P(U \in A \mid X = 1) \approx 0.12, \qquad P(U \in A \mid X = 0) \approx 0.88,
\end{align}
for the upper set $A = \{u : u^{(2)} = 1\}$. Hence $P(U \mid X = 1) \not\succeq_{st} P(U \mid X = 0)$, and the dominance is in fact reversed along the $U^{(2)}$ coordinate.
\end{example}
The intuition for the example can be described as follows. Conditioning on $X = 1$ pulls the distribution of $U^{(1)}$ upward, since $\pi$ is increasing in $u^{(1)}$. The strong negative dependence between $U^{(1)}$ and $U^{(2)}$ then drags the distribution of $U^{(2)}$ downward, against the direction of the dominance claim. Log-supermodularity of $P(U \mid Z = z)$ rules out exactly this kind of trade-off, ensuring that the $\{X=0\to X=1\}$-induced shift along one coordinate does not get reversed along another.
\fi
In the context of our ICU example, this means that the assumption would fail if unobserved severity markers traded off against one another -- e.g., if patients were sick along their respiratory axis \emph{or} their cardiac axis, but rarely both. As argued above, this is unlikely to be true for our ICU setting, in which severity tends to cluster across organ systems.

For grounding $Z$-dominance of \cref{prop:z-dom}, we discuss the logistic-Gaussian setting:
\begin{example}[$Z$-dominance in the logistic-Gaussian setting]\label{ex:z-dom-log-gauss}
Let $U, Z \in \R$ be jointly Gaussian with unit variances and correlation $\rho$, and let
\begin{align}
    P(X = 1 \mid Z, U) = \sigma(\alpha + \beta Z + \gamma U), \qquad \sigma(t) := (1 + e^{-t})^{-1},
\end{align}
with $\beta, \gamma > 0$. Since $U$ is univariate, the within-$u$ inequality in \cref{eq:across-zu} is vacuous and only the across-$z$-$u$ inequality \cref{eq:across-zu} must be checked. We work the $x = 0$ side. The cross derivative decomposes as
\begin{align}
    \frac{\partial^2 \log h^{(0)}(z, u)}{\partial u\, \partial z} = \underbrace{\frac{\partial^2 \log(1-\sigma)}{\partial u\, \partial z}}_{\text{collider}} + \underbrace{\frac{\partial^2 \log P(u \mid z)}{\partial u\, \partial z}}_{\text{bidirected}}.
\end{align}
A direct computation gives
\begin{align}
    \frac{\partial^2 \log(1 - \sigma(\alpha + \beta z + \gamma u))}{\partial u\, \partial z} = -\beta\gamma \cdot \sigma(\alpha + \beta z + \gamma u)(1 - \sigma(\alpha + \beta z + \gamma u)),
\end{align}
and, using $\log P(u \mid z) = -\tfrac{1}{2(1-\rho^2)}(u - \rho z)^2 + \mathrm{const}$,
\begin{align}
    \frac{\partial^2 \log P(u \mid z)}{\partial u\, \partial z} = \frac{\rho}{1 - \rho^2}.
\end{align}
The condition \cref{eq:across-zu} therefore becomes
\begin{align}
    \frac{\rho}{1 - \rho^2} \;\leq\; \beta\gamma \cdot \sigma(\alpha + \beta z + \gamma u)(1 - \sigma(\alpha + \beta z + \gamma u)). \label{eq:z-dom-log-gauss}
\end{align}
We note that the same inequality is obtained on the $x = 1$ side: in the logistic case $\partial_z \partial_u \log \sigma = \partial_z \partial_u \log(1 - \sigma) = -\beta\gamma\, \sigma(1-\sigma)$, so both chains in \cref{prop:z-dom} succeed under the same condition.

Two regions of $(u, z)$-space are illustrative. Near the decision boundary, where $\alpha + \beta z + \gamma u \approx 0$ and $\sigma(1-\sigma) = 1/4$, \cref{eq:z-dom-log-gauss} reduces to $\rho/(1 - \rho^2) \leq \beta\gamma/4$, permitting $\rho \lesssim 0.24$ when $\beta\gamma = 1$ and $\rho \lesssim 0.62$ when $\beta\gamma = 4$. Further into the tails, $\sigma(1-\sigma) \to 0$ and the condition tightens to $\rho \leq 0$.
\end{example}
The example shows that (iii) is genuinely weaker than $Z \ci U$: arbitrary positive marginal correlation is permitted, provided the propensity coefficients are sufficiently large. It also makes the trade-off explicit: increasing $\beta\gamma$ relaxes \cref{eq:z-dom-log-gauss}, while in a multivariate extension increasing $\gamma$ would make the within-$U$ condition harder to satisfy: the regime that satisfies both is one of strong propensity responsiveness paired with strong intra-$U$ marginal coupling.

The following examples illustrate how the chain \cref{eq:z-shift-x} of \cref{prop:z-dom} may fail:
\begin{example}[Interaction in $\pi$ at fixed $X$ may break one chain even under $Z \ci U$]\label{ex:zdom-fail-interaction}
Consider the SCM
\begin{align}
    Z &\sim \mathrm{Bern}(0.5), \quad U \sim \mathrm{Unif}[0, 1], \quad Z \ci U, \\
    X &\gets \mathrm{Bern}(\pi(Z, U)), \quad \pi(Z, U) = \begin{cases} U & \text{if } Z = 0 \\ U & \text{if } Z = 1, U < 0.5 \\ 1 & \text{if } Z = 1, U \geq 0.5. \end{cases}
\end{align}
A direct computation gives $\pi(z_0) = 0.5$, $\pi(z_1) = 0.625$, and
\begin{align}
    P(U > \tfrac{1}{2} \mid Z = z_1, X = 1) &= 0.80, & P(U > \tfrac{1}{2} \mid Z = z_0, X = 1) &= 0.75, \\
    P(U > \tfrac{1}{2} \mid Z = z_1, X = 0) &= 0, & P(U > \tfrac{1}{2} \mid Z = z_0, X = 0) &= 0.25.
\end{align}
Interestingly, the two-step $X$ and $Z$-dominance still holds 
$$P(U > \tfrac{1}{2} \mid z_0, X = 1) = 0.75 \geq 0 = P(U > \tfrac{1}{2} \mid z_1, X = 0),$$ but only via the $X = 0$ chain. The $X = 1$ chain fails ($0.75 \not\geq 0.80$ above), because the saturation of $\pi$ in the $(z_1, u \geq 0.5)$ corner makes (iii) fail at $X = 1$. The example illustrates that (iii) is genuinely a per-chain assumption: an interaction in the propensity at fixed $X$ can break one side without breaking the other, and \cref{prop:z-dom} only needs one $X=x$ to work.
\end{example}
Our further example illustrates why strong positive correlation between $Z$ and $U$ may break the $X$ and $Z$-dominance:
\begin{example}[Strong positive $Z$-$U$ correlation breaks both chains]\label{ex:zdom-fail-correlation}
Consider the SCM
\begin{align}
    Z &\sim \mathrm{Bern}(0.5), \\
    U \mid Z &\sim \begin{cases} \mathrm{Bern}(0.1) & \text{if } Z = 0 \\ \mathrm{Bern}(0.9) & \text{if } Z = 1 \end{cases} \\
    X &\gets \mathrm{Bern}(\pi(Z, U)), \quad \pi(Z, U) = \begin{cases} 0.2 & Z = 0, U = 0 \\ 0.4 & Z = 0, U = 1 \\ 0.3 & Z = 1, U = 0 \\ 0.5 & Z = 1, U = 1, \end{cases}
\end{align}
where $\pi$ is non-decreasing in each of $Z, U$. We compute $\pi(z_0) = 0.22$, $\pi(z_1) = 0.48$, and
\begin{align}
    P(U = 1 \mid Z = z_0, X = 1) \approx 0.18, \qquad P(U = 1 \mid Z = z_1, X = 0) \approx 0.87,
\end{align}
so $Z$-dominance fails: the smaller-$Z$ treated unit is stochastically smaller along $U$ than the larger-$Z$ untreated unit. The strong positive $Z$-$U$ marginal dependence concentrates joint mass on the $(0, 0)$ and $(1, 1)$ corners; (iii) fails for both $x = 0$ and $x = 1$, so neither chain is available. This is the second failure mode: if marginal $Z$-$U$ dependence is too strong relative to the responsiveness of $\pi$, the bidirected force overwhelms the collider force, and dominance doesn't hold.
\end{example}
\newpage
\section{\texorpdfstring{$\pi$}{Pi}-Dominance} \label{appendix:pi-dom}
In this appendix, we discuss the $\pi$-dominance strategy.
The same chain decomposition that produced \cref{prop:z-dom} applies with the propensity score $\pi(Z)$ replacing the full vector $Z$, giving a $\pi$-dominance result that compares two propensity levels $p' < p$ at fixed treatment.
\begin{proposition}[$\pi$-dominance]\label{prop:pi-dom}
Under assumptions of \cref{prop:z-match-multi}, and additionally:
\begin{enumerate}[label=(\roman*'), leftmargin=1.75em, nosep]
\setcounter{enumi}{2}
    \item For some $x \in \{0, 1\}$, defining $h^{(x)}_\pi(p, u) \coloneqq P(U = u, \pi(Z) = p, X = x)$, $\log h^{(x)}_\pi$ is twice differentiable and strictly positive, with
    \begin{align}
        \frac{\partial^2 \log h^{(x)}_\pi}{\partial u^{(j)}\, \partial u^{(k)}}(p, u) &\;\geq\; 0 \quad \text{for all } j \neq k, \quad
        \frac{\partial^2 \log h^{(x)}_\pi}{\partial p\, \partial u^{(j)}}(p, u) &\;\leq\; 0 \quad \text{for all } j.
    \end{align} \label{cond:pi-dom-cross} 
\end{enumerate}
Then, for $p' < p$,
\begin{align}
    P(U \mid \pi(Z) = p', X = 1) \;\succeq_{st}\; P(U \mid \pi(Z) = p, X = 0).
\end{align}
\end{proposition}
The proof is identical to that of \cref{prop:z-dom}, with the scalar $p$ replacing $z$ throughout: \cref{prop:pi-match-multi} provides the $\pi$-matching link of the chain, and \ref{cond:pi-dom-cross} provides the $\pi$-shift link via Karlin--Rinott.
The collider/bidirected interpretation of \cref{prop:z-dom} carries over: the within-$U$ inequality again asks that intra-$U$ marginal clustering survives the collider re-weighting, while the across-$\pi$-$U$ inequality asks that the collider channel overturn the marginal dependence between $\pi(Z)$ and $U$. The conceptual gain over $Z$-dominance is that the across-direction inequality is now a single one-dimensional cross-partial, regardless of $\dim(Z)$: the propensity score concentrates the structural assumption along the single dimension that actually drives treatment assignment. The within-$U$ inequality, by contrast, is not strictly weaker than its $Z$-dominance counterpart, since $h^{(x)}_\pi$ is a marginalization of $h^{(x)}$ over the level set $\{z : \pi(z) = p\}$, and log-supermodularity is not preserved under marginalization in dimensions higher than two; the conditions \ref{cond:z-dom-cross} and \ref{cond:pi-dom-cross} of \cref{prop:z-dom,prop:pi-dom} are therefore similar in spirit but logically independent (one does not imply the other).

\ifnum\treportflag=0
\newpage
\section*{NeurIPS Paper Checklist}

\begin{enumerate}

\item {\bf Claims}
    \item[] Question: Do the main claims made in the abstract and introduction accurately reflect the paper's contributions and scope?
    \item[] Answer: \answerYes{}
    \item[] Justification: The contributions listed in the introduction are substantiated: the formalization of the Confounder Detection via Treatment Intent (CDTI) design (\cref{sec:methods}), theoretical results on stochastic dominance for matching strategies (\cref{thm:stoch-dom}) and strategy dominance (\cref{thm:zdom-beats-zmatch}, \cref{thm:pi-beats-rand}), and the empirical validation (\cref{sec:experiments}).
    \item[] Guidelines:
    \begin{itemize}
        \item The answer \answerNA{} means that the abstract and introduction do not include the claims made in the paper.
        \item The abstract and/or introduction should clearly state the claims made, including the contributions made in the paper and important assumptions and limitations. A \answerNo{} or \answerNA{} answer to this question will not be perceived well by the reviewers. 
        \item The claims made should match theoretical and experimental results, and reflect how much the results can be expected to generalize to other settings. 
        \item It is fine to include aspirational goals as motivation as long as it is clear that these goals are not attained by the paper. 
    \end{itemize}

\item {\bf Limitations}
    \item[] Question: Does the paper discuss the limitations of the work performed by the authors?
    \item[] Answer: \answerYes{}
    \item[] Justification: 
    Limitations discussed explicitly at the end of the paper.
    \item[] Guidelines:
    \begin{itemize}
        \item The answer \answerNA{} means that the paper has no limitation while the answer \answerNo{} means that the paper has limitations, but those are not discussed in the paper. 
        \item The authors are encouraged to create a separate ``Limitations'' section in their paper.
        \item The paper should point out any strong assumptions and how robust the results are to violations of these assumptions (e.g., independence assumptions, noiseless settings, model well-specification, asymptotic approximations only holding locally). The authors should reflect on how these assumptions might be violated in practice and what the implications would be.
        \item The authors should reflect on the scope of the claims made, e.g., if the approach was only tested on a few datasets or with a few runs. In general, empirical results often depend on implicit assumptions, which should be articulated.
        \item The authors should reflect on the factors that influence the performance of the approach. For example, a facial recognition algorithm may perform poorly when image resolution is low or images are taken in low lighting. Or a speech-to-text system might not be used reliably to provide closed captions for online lectures because it fails to handle technical jargon.
        \item The authors should discuss the computational efficiency of the proposed algorithms and how they scale with dataset size.
        \item If applicable, the authors should discuss possible limitations of their approach to address problems of privacy and fairness.
        \item While the authors might fear that complete honesty about limitations might be used by reviewers as grounds for rejection, a worse outcome might be that reviewers discover limitations that aren't acknowledged in the paper. The authors should use their best judgment and recognize that individual actions in favor of transparency play an important role in developing norms that preserve the integrity of the community. Reviewers will be specifically instructed to not penalize honesty concerning limitations.
    \end{itemize}

\item {\bf Theory assumptions and proofs}
    \item[] Question: For each theoretical result, does the paper provide the full set of assumptions and a complete (and correct) proof?
    \item[] Answer: \answerYes{}
    \item[] Justification: All assumptions are clearly stated alongside the theoretical results in the main text (\cref{sec:methods}). Complete proofs for all propositions and theorems (\cref{thm:stoch-dom}, \cref{thm:zdom-beats-zmatch}, \cref{thm:pi-beats-rand}) are provided in \cref{app:proofs}.
    \item[] Guidelines:
    \begin{itemize}
        \item The answer \answerNA{} means that the paper does not include theoretical results. 
        \item All the theorems, formulas, and proofs in the paper should be numbered and cross-referenced.
        \item All assumptions should be clearly stated or referenced in the statement of any theorems.
        \item The proofs can either appear in the main paper or the supplemental material, but if they appear in the supplemental material, the authors are encouraged to provide a short proof sketch to provide intuition. 
        \item Inversely, any informal proof provided in the core of the paper should be complemented by formal proofs provided in appendix or supplemental material.
        \item Theorems and Lemmas that the proof relies upon should be properly referenced. 
    \end{itemize}

    \item {\bf Experimental result reproducibility}
    \item[] Question: Does the paper fully disclose all the information needed to reproduce the main experimental results of the paper to the extent that it affects the main claims and/or conclusions of the paper (regardless of whether the code and data are provided or not)?
    \item[] Answer: \answerYes{}
    \item[] Justification: All data sources (MIMIC-III, AUMCdb, SICdb) and synthetic data generation processes are described in the main text (\cref{sec:experiments}). Full code, including data generation scripts and plotting composers, is provided in an anonymized repository.
    \item[] Guidelines:
    \begin{itemize}
        \item The answer \answerNA{} means that the paper does not include experiments.
        \item If the paper includes experiments, a \answerNo{} answer to this question will not be perceived well by the reviewers: Making the paper reproducible is important, regardless of whether the code and data are provided or not.
        \item If the contribution is a dataset and\slash or model, the authors should describe the steps taken to make their results reproducible or verifiable. 
        \item Depending on the contribution, reproducibility can be accomplished in various ways. For example, if the contribution is a novel architecture, describing the architecture fully might suffice, or if the contribution is a specific model and empirical evaluation, it may be necessary to either make it possible for others to replicate the model with the same dataset, or provide access to the model. In general. releasing code and data is often one good way to accomplish this, but reproducibility can also be provided via detailed instructions for how to replicate the results, access to a hosted model (e.g., in the case of a large language model), releasing of a model checkpoint, or other means that are appropriate to the research performed.
        \item While NeurIPS does not require releasing code, the conference does require all submissions to provide some reasonable avenue for reproducibility, which may depend on the nature of the contribution. For example
        \begin{enumerate}
            \item If the contribution is primarily a new algorithm, the paper should make it clear how to reproduce that algorithm.
            \item If the contribution is primarily a new model architecture, the paper should describe the architecture clearly and fully.
            \item If the contribution is a new model (e.g., a large language model), then there should either be a way to access this model for reproducing the results or a way to reproduce the model (e.g., with an open-source dataset or instructions for how to construct the dataset).
            \item We recognize that reproducibility may be tricky in some cases, in which case authors are welcome to describe the particular way they provide for reproducibility. In the case of closed-source models, it may be that access to the model is limited in some way (e.g., to registered users), but it should be possible for other researchers to have some path to reproducing or verifying the results.
        \end{enumerate}
    \end{itemize}

\item {\bf Open access to data and code}
    \item[] Question: Does the paper provide open access to the data and code, with sufficient instructions to faithfully reproduce the main experimental results, as described in supplemental material?
    \item[] Answer: \answerYes{}
    \item[] Justification: The code is publicly available in an anonymized repository. The ICU databases (MIMIC-III, AUMCdb, SICdb) are publicly available to credentialed researchers. The paper and code repository describe all preprocessing steps necessary to reproduce the analyses.
    \item[] Guidelines:
    \begin{itemize}
        \item The answer \answerNA{} means that paper does not include experiments requiring code.
        \item Please see the NeurIPS code and data submission guidelines (\url{https://neurips.cc/public/guides/CodeSubmissionPolicy}) for more details.
        \item While we encourage the release of code and data, we understand that this might not be possible, so \answerNo{} is an acceptable answer. Papers cannot be rejected simply for not including code, unless this is central to the contribution (e.g., for a new open-source benchmark).
        \item The instructions should contain the exact command and environment needed to run to reproduce the results. See the NeurIPS code and data submission guidelines (\url{https://neurips.cc/public/guides/CodeSubmissionPolicy}) for more details.
        \item The authors should provide instructions on data access and preparation, including how to access the raw data, preprocessed data, intermediate data, and generated data, etc.
        \item The authors should provide scripts to reproduce all experimental results for the new proposed method and baselines. If only a subset of experiments are reproducible, they should state which ones are omitted from the script and why.
        \item At submission time, to preserve anonymity, the authors should release anonymized versions (if applicable).
        \item Providing as much information as possible in supplemental material (appended to the paper) is recommended, but including URLs to data and code is permitted.
    \end{itemize}

\item {\bf Experimental setting/details}
    \item[] Question: Does the paper specify all the training and test details (e.g., data splits, hyperparameters, how they were chosen, type of optimizer) necessary to understand the results?
    \item[] Answer: \answerYes{}
    \item[] Justification: Experimental details are explicitly provided in \cref{sec:experiments}.
    \item[] Guidelines:
    \begin{itemize}
        \item The answer \answerNA{} means that the paper does not include experiments.
        \item The experimental setting should be presented in the core of the paper to a level of detail that is necessary to appreciate the results and make sense of them.
        \item The full details can be provided either with the code, in appendix, or as supplemental material.
    \end{itemize}

\item {\bf Experiment statistical significance}
    \item[] Question: Does the paper report error bars suitably and correctly defined or other appropriate information about the statistical significance of the experiments?
    \item[] Answer: \answerYes{}
    \item[] Justification: Where appropriate, confidence intervals are reported.
    \item[] Guidelines:
    \begin{itemize}
        \item The answer \answerNA{} means that the paper does not include experiments.
        \item The authors should answer \answerYes{} if the results are accompanied by error bars, confidence intervals, or statistical significance tests, at least for the experiments that support the main claims of the paper.
        \item The factors of variability that the error bars are capturing should be clearly stated (for example, train/test split, initialization, random drawing of some parameter, or overall run with given experimental conditions).
        \item The method for calculating the error bars should be explained (closed form formula, call to a library function, bootstrap, etc.)
        \item The assumptions made should be given (e.g., Normally distributed errors).
        \item It should be clear whether the error bar is the standard deviation or the standard error of the mean.
        \item It is OK to report 1-sigma error bars, but one should state it. The authors should preferably report a 2-sigma error bar than state that they have a 96\% CI, if the hypothesis of Normality of errors is not verified.
        \item For asymmetric distributions, the authors should be careful not to show in tables or figures symmetric error bars that would yield results that are out of range (e.g., negative error rates).
        \item If error bars are reported in tables or plots, the authors should explain in the text how they were calculated and reference the corresponding figures or tables in the text.
    \end{itemize}

\item {\bf Experiments compute resources}
    \item[] Question: For each experiment, does the paper provide sufficient information on the computer resources (type of compute workers, memory, time of execution) needed to reproduce the experiments?
    \item[] Answer: \answerYes{}
    \item[] Justification: Information regarding compute resources is provided in the supplementary material and code repository README. Most experiments were run on a MacBook Pro M3 (Tahoe 26.2) with under 24 hours of total computation.
    \item[] Guidelines:
    \begin{itemize}
        \item The answer \answerNA{} means that the paper does not include experiments.
        \item The paper should indicate the type of compute workers CPU or GPU, internal cluster, or cloud provider, including relevant memory and storage.
        \item The paper should provide the amount of compute required for each of the individual experimental runs as well as estimate the total compute. 
        \item The paper should disclose whether the full research project required more compute than the experiments reported in the paper (e.g., preliminary or failed experiments that didn't make it into the paper). 
    \end{itemize}
    
\item {\bf Code of ethics}
    \item[] Question: Does the research conducted in the paper conform, in every respect, with the NeurIPS Code of Ethics \url{https://neurips.cc/public/EthicsGuidelines}?
    \item[] Answer: \answerYes{}
    \item[] Justification: The research complies with the NeurIPS Code of Ethics.
    \item[] Guidelines:
    \begin{itemize}
        \item The answer \answerNA{} means that the authors have not reviewed the NeurIPS Code of Ethics.
        \item If the authors answer \answerNo, they should explain the special circumstances that require a deviation from the Code of Ethics.
        \item The authors should make sure to preserve anonymity (e.g., if there is a special consideration due to laws or regulations in their jurisdiction).
    \end{itemize}

\item {\bf Broader impacts}
    \item[] Question: Does the paper discuss both potential positive societal impacts and negative societal impacts of the work performed?
    \item[] Answer: \answerYes{}
    \item[] Justification: Potential impacts are discussed. The work aims to positively impact the reliability of observational causal inference by providing a method for identifying unmeasured confounders. We do not foresee direct ethical risks or negative societal impacts from this work.
    \item[] Guidelines:
    \begin{itemize}
        \item The answer \answerNA{} means that there is no societal impact of the work performed.
        \item If the authors answer \answerNA{} or \answerNo, they should explain why their work has no societal impact or why the paper does not address societal impact.
        \item Examples of negative societal impacts include potential malicious or unintended uses (e.g., disinformation, generating fake profiles, surveillance), fairness considerations (e.g., deployment of technologies that could make decisions that unfairly impact specific groups), privacy considerations, and security considerations.
        \item The conference expects that many papers will be foundational research and not tied to particular applications, let alone deployments. However, if there is a direct path to any negative applications, the authors should point it out. For example, it is legitimate to point out that an improvement in the quality of generative models could be used to generate Deepfakes for disinformation. On the other hand, it is not needed to point out that a generic algorithm for optimizing neural networks could enable people to train models that generate Deepfakes faster.
        \item The authors should consider possible harms that could arise when the technology is being used as intended and functioning correctly, harms that could arise when the technology is being used as intended but gives incorrect results, and harms following from (intentional or unintentional) misuse of the technology.
        \item If there are negative societal impacts, the authors could also discuss possible mitigation strategies (e.g., gated release of models, providing defenses in addition to attacks, mechanisms for monitoring misuse, mechanisms to monitor how a system learns from feedback over time, improving the efficiency and accessibility of ML).
    \end{itemize}
    
\item {\bf Safeguards}
    \item[] Question: Does the paper describe safeguards that have been put in place for responsible release of data or models that have a high risk for misuse (e.g., pre-trained language models, image generators, or scraped datasets)?
    \item[] Answer: \answerNA{}
    \item[] Justification: The paper introduces a novel observational study design. It does not release new generative models or datasets, posing no such risks.
    \item[] Guidelines:
    \begin{itemize}
        \item The answer \answerNA{} means that the paper poses no such risks.
        \item Released models that have a high risk for misuse or dual-use should be released with necessary safeguards to allow for controlled use of the model, for example by requiring that users adhere to usage guidelines or restrictions to access the model or implementing safety filters. 
        \item Datasets that have been scraped from the Internet could pose safety risks. The authors should describe how they avoided releasing unsafe images.
        \item We recognize that providing effective safeguards is challenging, and many papers do not require this, but we encourage authors to take this into account and make a best faith effort.
    \end{itemize}

\item {\bf Licenses for existing assets}
    \item[] Question: Are the creators or original owners of assets (e.g., code, data, models), used in the paper, properly credited and are the license and terms of use explicitly mentioned and properly respected?
    \item[] Answer: \answerYes{}
    \item[] Justification: All assets used (MIMIC-III, AUMCdb, SICdb) are properly cited in the Introduction and their data usage terms are appropriately respected.
    \item[] Guidelines:
    \begin{itemize}
        \item The answer \answerNA{} means that the paper does not use existing assets.
        \item The authors should cite the original paper that produced the code package or dataset.
        \item The authors should state which version of the asset is used and, if possible, include a URL.
        \item The name of the license (e.g., CC-BY 4.0) should be included for each asset.
        \item For scraped data from a particular source (e.g., website), the copyright and terms of service of that source should be provided.
        \item If assets are released, the license, copyright information, and terms of use in the package should be provided. For popular datasets, \url{paperswithcode.com/datasets} has curated licenses for some datasets. Their licensing guide can help determine the license of a dataset.
        \item For existing datasets that are re-packaged, both the original license and the license of the derived asset (if it has changed) should be provided.
        \item If this information is not available online, the authors are encouraged to reach out to the asset's creators.
    \end{itemize}

\item {\bf New assets}
    \item[] Question: Are new assets introduced in the paper well documented and is the documentation provided alongside the assets?
    \item[] Answer: \answerYes{}
    \item[] Justification: The anonymized code includes a README file with instruction explaining how to use the code.
    \item[] Guidelines:
    \begin{itemize}
        \item The answer \answerNA{} means that the paper does not release new assets.
        \item Researchers should communicate the details of the dataset\slash code\slash model as part of their submissions via structured templates. This includes details about training, license, limitations, etc. 
        \item The paper should discuss whether and how consent was obtained from people whose asset is used.
        \item At submission time, remember to anonymize your assets (if applicable). You can either create an anonymized URL or include an anonymized zip file.
    \end{itemize}

\item {\bf Crowdsourcing and research with human subjects}
    \item[] Question: For crowdsourcing experiments and research with human subjects, does the paper include the full text of instructions given to participants and screenshots, if applicable, as well as details about compensation (if any)?
    \item[] Answer: \answerNA{}
    \item[] Justification: No crowdsourcing of human subjects involved.
    \item[] Guidelines:
    \begin{itemize}
        \item The answer \answerNA{} means that the paper does not involve crowdsourcing nor research with human subjects.
        \item Including this information in the supplemental material is fine, but if the main contribution of the paper involves human subjects, then as much detail as possible should be included in the main paper. 
        \item According to the NeurIPS Code of Ethics, workers involved in data collection, curation, or other labor should be paid at least the minimum wage in the country of the data collector. 
    \end{itemize}

\item {\bf Institutional review board (IRB) approvals or equivalent for research with human subjects}
    \item[] Question: Does the paper describe potential risks incurred by study participants, whether such risks were disclosed to the subjects, and whether Institutional Review Board (IRB) approvals (or an equivalent approval/review based on the requirements of your country or institution) were obtained?
    \item[] Answer: \answerNA{}
    \item[] Justification: The study relies exclusively on publicly available, completely de-identified electronic health records (MIMIC-III, AUMCdb, SICdb). Because this data is fully anonymized and retrospective, it poses no risk to individuals and is exempt from requiring Institutional Review Board (IRB) approval.
    \item[] Guidelines:
    \begin{itemize}
        \item The answer \answerNA{} means that the paper does not involve crowdsourcing nor research with human subjects.
        \item Depending on the country in which research is conducted, IRB approval (or equivalent) may be required for any human subjects research. If you obtained IRB approval, you should clearly state this in the paper. 
        \item We recognize that the procedures for this may vary significantly between institutions and locations, and we expect authors to adhere to the NeurIPS Code of Ethics and the guidelines for their institution. 
        \item For initial submissions, do not include any information that would break anonymity (if applicable), such as the institution conducting the review.
    \end{itemize}

\item {\bf Declaration of LLM usage}
    \item[] Question: Does the paper describe the usage of LLMs if it is an important, original, or non-standard component of the core methods in this research? Note that if the LLM is used only for writing, editing, or formatting purposes and does \emph{not} impact the core methodology, scientific rigor, or originality of the research, declaration is not required.
    \item[] Answer: \answerNA{}
    \item[] Justification: Language models were used for text polishing and checking correctness, but not for developing core methodology.
    \item[] Guidelines:
    \begin{itemize}
        \item The answer \answerNA{} means that the core method development in this research does not involve LLMs as any important, original, or non-standard components.
        \item Please refer to our LLM policy in the NeurIPS handbook for what should or should not be described.
    \end{itemize}

\end{enumerate}
\fi

\end{document}